\documentclass{article}
\usepackage{amssymb}
\usepackage{amsfonts}
\usepackage{amsmath}

\newtheorem{theorem}{Theorem}

\newtheorem{corollary}[theorem]{Corollary}

\newtheorem{definition}[theorem]{Definition}
\newtheorem{example}[theorem]{Example}

\newtheorem{lemma}[theorem]{Lemma}

\newtheorem{proposition}[theorem]{Proposition}

\newenvironment{proof}[1][Proof]{\noindent\textbf{#1.} }{\ \rule{0.5em}{0.5em}}
\input{tcilatex}

\begin{document}

\title{On the codes over the $Z_{3}+vZ_{3}+v^{2}Z_{3}$}
\author{Abdullah Dertli$^{a}$, Yasemin Cengellenmis$^{b}$, Senol Eren$^{a}$
\and $\left( a\right) $ Ondokuz May\i s University, Faculty of Arts and
Sciences, \and Mathematics Department, Samsun, Turkey \and %
abdullah.dertli@gmail.com, seren@omu.edu.tr \and $\left( b\right) $ Trakya
University, Faculty of Arts and Sciences, \and Mathematics Department,
Edirne, Turkey \and ycengellenmis@gmail.com}
\maketitle

\begin{abstract}
In this paper, we study the structure of cyclic, quasi-cyclic, constacyclic
codes and their skew codes over the finite ring $R=Z_{3}+vZ_{3}+v^{2}Z_{3}$, 
$v^{3}=v.$ The Gray images of cyclic, quasi-cyclic, skew cyclic, skew
quasi-cyclic and skew constacyclic codes over $R$ are obtained.\ A necessary
and sufficient condition for cyclic (negacyclic) codes over $R$ that
contains its dual has been given. The parameters of quantum error correcting
codes are obtained from both cyclic and negacyclic codes over $R$. It is
given some examples. Firstly, quasi-constacyclic and skew quasi-constacyclic
codes are introduced. By giving two Hermitian product, it is investigated
their duality. A sufficient condition for 1-generator skew
quasi-constacyclic codes to be free is determined.
\end{abstract}

\section{Introduction}

At the beginning, a lot of research on error correcting codes are
concentrated on codes over finite fields. Since the revelation in 1994 [6],
there has been a lot of interest in codes over finite rings. The structure
of certain type of codes over many rings are determined such as negacyclic,
cyclic, quasi-cyclic, consta cyclic codes in [11,19,22,25,26,27,28,33]. Many
methods and many approaches are applied to produce certain types of codes
with good parameters and properties.

Some authors generalized the notion of cyclic, quasi-cyclic and constacyclic
codes by using generator polynomials in skew polynomial rings
[7,8,9,10,12,15,16,\newline
18,21,24,29,30].

Moreover, in [5] Calderbank et al. gave a way to construct quantum error
correcting codes from classical error correcting codes, although the theory
of quantum error correcting codes has striking differences from the theory
of classical error correcting codes. Many good quantum codes have been
constructed by using classical cyclic codes over finite fields or finite
rings with self orthogonal (or dual containing) properties in
[2,3,4,13,14,17,20,23,31,32].

In [1] they introduced the finite ring $R=Z_{3}[v]/\left\langle
v^{3}-v\right\rangle .$ They studied the structure of this ring. The
algebraic structure of cyclic and dual codes was also studied. A MacWilliams
type identity was established.

In this paper, it is given some definitions. By giving the duality of codes
via inner product, it is shown that $C$ is self orthogonal codes over $R$,
so is $\phi \left( C\right) $, where $\phi $ is a Gray map.

The Gray images of cyclic and quasi-cyclic codes over R are obtained. A
linear codes over R is represented by means of three ternary codes and the
generator matrix is given.

After a cyclic (negacyclic) codes over R is represented via cyclic
(negacyclic) codes over $Z_{3},$ it is determined the dual of cyclic
(negacyclic) codes. A necessary and sufficient condition for cyclic
(negacyclic) code over R that contains its dual is given. The parameters of
quantum error correcting codes are obtained from both cyclic and negacyclic
codes over R. As a last, some examples are given about quantum error
correcting codes.

When n is odd, it is defined the $\lambda $-constacyclic codes over R where $%
\lambda $ is unit. A constacyclic codes is represented by means of either a
cyclic code or a negacyclic code of length n.

It is found the nontrivial automorphism $\theta $ on the ring R. By using
this automorphism, the skew cyclic, skew quasi-cyclic and skew constacyclic
codes over R are introduced. The number of distinct skew cyclic codes over R
is given. The Gray image of skew codes is obtained.

Firstly, quasi-constacyclic and skew quasi-constacyclic codes over R are
introduced. By using two Hermitian product, it is investigated the duality
about quasi-constacyclic and skew quasi-constacyclic codes over R. The Gray
image of skew quasi-constacyclic codes over R is determined. A sufficient
condition for 1-generator skew quasi-constacyclic to be free is determined.

\section{Preliminaries}

Suppose $R=Z_{3}+vZ_{3}+v^{2}Z_{3}$ where $v^{3}=v$ and $Z_{3}=\left\{
0,1,2\right\} $. $R$ is a finite commutative ring with $27$ elements. This
ring is a semi local ring with three maximal ideals. $R$\ is a principal
ideal ring and not finite chain ring. The units of the ring are $%
1,2,1+v^{2},1+v+2v^{2},1+2v+2v^{2},2+v+v^{2},2+2v+v^{2},2+2v^{2}.$The
maximal ideals,

\begin{eqnarray*}
\left\langle v\right\rangle &=&\left\langle 2v\right\rangle =\left\langle
v^{2}\right\rangle =\left\langle 2v^{2}\right\rangle \\
&=&\left\{ 0,v,2v,v^{2},2v^{2},v+v^{2},v+2v^{2},2v+v^{2},2v+2v^{2}\right\} \\
\left\langle 1+v\right\rangle &=&\left\langle 2+2v\right\rangle
=\left\langle 1+2v+v^{2}\right\rangle =\left\langle 2+v+2v^{2}\right\rangle
\end{eqnarray*}

\begin{eqnarray*}
&=&\{0,1+v,2+2v,v+v^{2},2v+2v^{2},1+2v+v^{2},1+2v^{2},2+v^{2}, \\
&&2+v+2v^{2}\} \\
\left\langle 1+v+v^{2}\right\rangle &=&\left\langle 1+2v\right\rangle
=\left\langle 2+v\right\rangle =\left\langle 2+2v+2v^{2}\right\rangle \\
&=&\{0,2+v,1+2v,2v+v^{2},v+2v^{2},2+v^{2},1+2v^{2},2+2v+2v^{2}, \\
&&1+v+v^{2}\}
\end{eqnarray*}

The other ideals,%
\begin{eqnarray*}
\left\langle 0\right\rangle &=&\{0\} \\
\left\langle 1\right\rangle &=&\left\langle 2\right\rangle =\left\langle
1+v^{2}\right\rangle =\left\langle 1+v+2v^{2}\right\rangle =\left\langle
1+2v+2v^{2}\right\rangle =\left\langle 2+v+v^{2}\right\rangle \\
&=&\left\langle 2+2v+v^{2}\right\rangle =\left\langle 2+2v^{2}\right\rangle
=R \\
\left\langle 1+2v^{2}\right\rangle &=&\left\langle 2+v^{2}\right\rangle
=\{0,2+v^{2},1+2v^{2}\} \\
\left\langle v+v^{2}\right\rangle &=&\left\langle 2v+2v^{2}\right\rangle
=\{0,v+v^{2},2v+2v^{2}\} \\
\left\langle v+2v^{2}\right\rangle &=&\left\langle 2v+v^{2}\right\rangle
=\{0,v+2v^{2},2v+v^{2}\}
\end{eqnarray*}%
A linear code $C$ over $R$ length $n$ is a $R-$submodule of $R^{n}$. An
element of $C$ is called a codeword.

For any $x=\left( x_{0},x_{1},...,x_{n-1}\right) $, $y=\left(
y_{0},y_{1},...,y_{n-1}\right) $ the inner product is defined as

\begin{equation*}
x.y=\sum_{i=0}^{n-1}x_{i}y_{i}
\end{equation*}

If $x.y=0$ then $x$ and $y$ are said to be orthogonal. Let $C$ be linear
code of length $n$ over $R$, the dual code of $C$ 
\begin{equation*}
C^{\perp }=\left\{ x:\forall y\in C,x.y=0\right\}
\end{equation*}%
which is also a linear code over $R$ of length $n$. A code $C$ is self
orthogonal if $C\subseteq C^{\perp }$ and self dual if $C=C^{\perp }$.

A code $C$ over $R$ is a linear code with the property that if $c=\left(
c_{0},c_{1},...,c_{n-1}\right) \in C$ then $\sigma \left( C\right) =\left(
c_{n-1},c_{0},...,c_{n-2}\right) \in C.$ A subset $C$ of $R^{n}$ is a linear
cyclic code of length $n$ iff it is polynomial representation is an ideal of 
$R\left[ x\right] /\left\langle x^{n}-1\right\rangle .$

A code $C$ over $R$ is a linear code with the property that if $c=\left(
c_{0},c_{1},...,c_{n-1}\right) \in C$ then $\nu \left( C\right) =\left(
\lambda c_{n-1},c_{0},...,c_{n-2}\right) \in C$ where $\lambda $ is a unit
element of $R.$ A subset $C$ of $R^{n}$ is a linear $\lambda $-constacyclic
code of length $n$ iff it is polynomial representation is an ideal of $R%
\left[ x\right] /\left\langle x^{n}-\lambda \right\rangle .$

A code $C$ over $R$ is a linear code with the property that if $c=\left(
c_{0},c_{1},...,c_{n-1}\right) \in C$ then $\eta \left( C\right) =\left(
-c_{n-1},c_{0},...,c_{n-2}\right) \in C.$ A subset $C$ of $R^{n}$ is a
linear negacyclic code of length $n$ iff it is polynomial representation is
an ideal of $R\left[ x\right] /\left\langle x^{n}+1\right\rangle .$

Let $C$ be code over $Z_{3}$ of length $n$ and $\acute{c}=\left( \acute{c}%
_{0},\acute{c}_{1},...,\acute{c}_{n-1}\right) $ be a codeword of $C.$ The
Hamming weight of $\acute{c}$ is defined as $w_{H}\left( \acute{c}\right)
=\dsum\limits_{i=0}^{n-1}w_{H}\left( \acute{c}_{i}\right) $ where $%
w_{H}\left( \acute{c}_{i}\right) =1$ if $\acute{c}_{i}\neq 0$ and $%
w_{H}\left( \acute{c}_{i}\right) =0$ if $\acute{c}_{i}=0.$ Hamming distance
of $C$ is defined as $d_{H}\left( C\right) =\min d_{H}\left( c,\acute{c}%
\right) ,$ where for any $\acute{c}\in C,$ $c\neq \acute{c}$ and $%
d_{H}\left( c,\acute{c}\right) $ is Hamming distance between two codewords
with $d_{H}\left( c,\acute{c}\right) =w_{H}\left( c-\acute{c}\right) .$

Let $a\in Z_{3}^{3n}$ with $a=\left( a_{0},a_{1},...,a_{3n-1}\right) =\left(
a^{\left( 0\right) }\left\vert a^{\left( 1\right) }\right\vert a^{\left(
2\right) }\right) ,$ $a^{\left( i\right) }\in Z_{3}^{n}$ for $i=0,1,2.$ Let $%
\varphi $ be a map from $Z_{3}^{3n}$ to $Z_{3}^{3n}$ given by $\varphi
\left( a\right) =\left( \sigma \left( a^{\left( 0\right) }\right) \left\vert
\sigma \left( a^{\left( 1\right) }\right) \right\vert \sigma \left(
a^{\left( 2\right) }\right) \right) $ where $\sigma $ is a cyclic shift from 
$Z_{3}^{n}$ to $Z_{3}^{n}$ given by $\sigma \left( a^{\left( i\right)
}\right) =((a^{\left( i,n-1\right) }),(a^{\left( i,0\right) }),(a^{\left(
i,1\right) })$\newline
$,...,(a^{\left( i,n-2\right) }))$ for every $a^{\left( i\right)
}=(a^{(i,0)},...,a^{\left( i,n-1\right) })$ where $a^{\left( i,j\right) }\in
Z_{3}$, $j=0,1,...,n-1.$ A code of length $3n$ over $Z_{3}$ is said to be
quasi cyclic code of index $3$ if $\varphi \left( C\right) =C.$

Let $n=sl$. A quasi-cyclic code $C$ over $R$ of length $n$ and index $l$ is
a linear code with the property that if

$e=\left(
e_{0,0},...,e_{0,l-1},e_{1,0},...,e_{1,l-1},...,e_{s-1,0},...,e_{s-1,l-1}%
\right) \in C$, then $\tau _{s,l}\left( e\right) =\left(
e_{s-1,0,...,}e_{s-1,l-1},e_{0,0},...,e_{0,l-1},...,e_{s-2,0},...,e_{s-2,l-1}\right) \in C 
$.

Let $a\in Z_{3}^{3n}$ with $a=\left( a_{0},a_{1},...,a_{3n-1}\right) =\left(
a^{\left( 0\right) }\left\vert a^{\left( 1\right) }\right\vert a^{\left(
2\right) }\right) ,$ $a^{\left( i\right) }\in Z_{3}^{n}$, for $i=0,1,2.$ Let 
$\Gamma $ be a map from $Z_{3}^{3n}$ to $Z_{3}^{3n}$ given by

\begin{equation*}
\Gamma \left( a\right) =\left( \mu \left( a^{\left( 0\right) }\right)
\left\vert \mu \left( a^{\left( 1\right) }\right) \right\vert \mu \left(
a^{\left( 2\right) }\right) \right)
\end{equation*}%
where $\mu $ is the map from $Z_{3}^{n}$ to $Z_{3}^{n}$ given by

\begin{equation*}
\mu \left( a^{\left( i\right) }\right) =((a^{\left( i,s-1\right)
}),(a^{\left( i,0\right) }),...,(a^{\left( i,s-2\right) }))
\end{equation*}%
for every $a^{\left( i\right) }=\left( a^{\left( i,0\right) },...,a^{\left(
i,s-1\right) }\right) $ where $a^{\left( i,j\right) }\in Z_{3}^{l}$, $%
j=0,1,...,s-1$ and $n=sl.$ A code of length $3n$ over $Z_{3}$ is said to be $%
l-$quasi cyclic code of index $3$ if $\Gamma \left( C\right) =C.$

\section{Gray Map and Gray Images of Cyclic and Quasi-cyclic Codes Over R}

In [1], the Gray map is defined as follows%
\begin{eqnarray*}
\phi &:&R\rightarrow Z_{3}^{3} \\
\phi (a+vb+v^{2}c) &=&(a,a+b+c,a+2b+c)
\end{eqnarray*}

Let $C$ be a linear code over $R$ of length $n$. For any codeword $c=\left(
c_{0},...,c_{n-1}\right) $ the Lee weight of $c$ is defined as $w_{L}\left(
c\right) =\dsum\limits_{i=0}^{n-1}w_{L}\left( c_{i}\right) $ and the Lee
distance of $C$ is defined as $d_{L}\left( C\right) =\min d_{L}\left( c,%
\acute{c}\right) ,$ where for any $\acute{c}\in C,$ $c\neq \acute{c}$ and $%
d_{L}\left( c,\acute{c}\right) $ is Lee distance between two codewords with $%
d_{L}\left( c,\acute{c}\right) =w_{L}\left( c-\acute{c}\right) .$ Gray map $%
\phi $ can be extended to map from $R^{n}$ to $Z_{3}^{3n}.$

\begin{theorem}
The Gray map $\phi $ is a weight preserving map from $\left( R^{n},\text{Lee
weight}\right) $ to $\left( Z_{3}^{3n},\text{Hamming weight}\right) .$
Moreover it is an isometry from $R^{n}$ to $Z_{3}^{3n}$.
\end{theorem}

\begin{theorem}
If $C$ is an $\left[ n,k,d_{L}\right] $ linear codes over $R$ then $\phi
\left( C\right) $ is a $\left[ 3n,k,d_{H}\right] $ linear codes over $Z_{3},$
where $d_{H}=d_{L}.$
\end{theorem}

\begin{proof}
Let $x=a_{1}+vb_{1}+v^{2}c_{1},$ $y=a_{2}+vb_{2}+v^{2}c_{2}\in R,\alpha \in
Z_{3\text{ }}$ then

$\phi \left( x+y\right) =\phi \left( a_{1}+a_{2}+v\left( b_{1}+b_{2}\right)
+v^{2}\left( c_{1}+c_{2}\right) \right) $

$\qquad
=(a_{1}+a_{2},a_{1}+a_{2}+b_{1}+b_{2}+c_{1}+c_{2},a_{1}+a_{2}+2(b_{1}+b_{2})+c_{1}+c_{2}) 
$

$\qquad
=(a_{1},a_{1}+b_{1}+c_{1},a_{1}+2b_{1}+c_{1})+(a_{2},a_{2}+b_{2}+c_{2},a_{2}+2b_{2}+c_{2}) 
$

$\qquad =\phi \left( x\right) +\phi \left( y\right) $

$\phi \left( \alpha x\right) =\phi \left( \alpha a_{1}+v\alpha
b_{1}+v^{2}\alpha c_{1}\right) $

$\qquad =(\alpha a_{1},\alpha a_{1}+\alpha b_{1}+\alpha c_{1},\alpha
a_{1}+2\alpha b_{1}+\alpha c_{1})$

$\qquad =\alpha (a_{1},a_{1}+b_{1}+c_{1},a_{1}+2b_{1}+c_{1})$

$\qquad =\alpha \phi \left( x\right) $\ \ \ \ \ \ \ 

so $\phi $ is\textit{\ }linear\textit{. }As\textit{\ }$\phi $ is bijective
then $\left\vert C\right\vert =\left\vert \phi \left( C\right) \right\vert $%
. From theorem $1$ we have $d_{H}=d_{L}.$
\end{proof}

\begin{theorem}
If $C$ is self orthogonal, so is $\phi \left( C\right) .$
\end{theorem}

\begin{proof}
Let $x=a_{1}+vb_{1}+v^{2}c_{1},$ $y=a_{2}+vb_{2}+v^{2}c_{2}$ where $%
a_{1},b_{1},c_{1},$ $a_{2},b_{2},c_{2}\in Z_{3}$.

From $%
x.y=a_{1}a_{2}+v(a_{1}b_{2}+b_{1}a_{2}+b_{1}c_{2}+c_{1}b_{2})+v^{2}(a_{1}c_{2}+b_{1}b_{2}+c_{1}a_{2}+c_{1}c_{2}) 
$, if$\ C$ is self orthogonal,so we have $%
a_{1}a_{2}=0,a_{1}b_{2}+b_{1}a_{2}+b_{1}c_{2}+c_{1}b_{2}=0,a_{1}c_{2}+b_{1}b_{2}+c_{1}a_{2}+c_{1}c_{2}=0 
$. From

$\phi \left( x\right) .\phi \left( y\right)
=(a_{1},a_{1}+b_{1}+c_{1},a_{1}+2b_{1}+c_{1})(a_{2},a_{2}+b_{2}+c_{2},a_{2}+2b_{2}+c_{2}) 
$

\ \ \ \ \ \ \ \ \ \ \ \ \ \ \ \ $\
=a_{1}a_{2}+a_{1}a_{2}+a_{1}b_{2}+a_{1}c_{2}+b_{1}a_{2}+b_{1}b_{2}+b_{1}c_{2}+c_{1}a_{2}+c_{1}b_{2}+c_{1}c_{2}+a_{1}a_{2}+2(a_{1}b_{2}+b_{1}a_{2}+b_{1}c_{2}+c_{1}b_{2})+a_{1}c_{2}+b_{1}b_{2}+c_{1}a_{2}+c_{1}c_{2}=0 
$

Therefore, we have $\phi \left( C\right) $ is self orthogonal.
\end{proof}

Note that $\phi \left( C\right) ^{\bot }=\phi \left( C^{\bot }\right) .$
Moreover, if $C$ is self-dual, so is $\phi \left( C\right) .$

\begin{proposition}
Let $\phi $ the Gray map from $R^{n}$ to $Z_{3}^{3n}$, let $\sigma $ be
cyclic shift and let $\varphi $ be a map as in the preliminaries. Then $\phi
\sigma =\varphi \phi .$
\end{proposition}

\begin{proof}
Let $r_{i}=a_{i}+vb_{i}+v^{2}c_{i}$ be the elements of $R$ for $%
i=0,1,....,n-1.$ We have $\sigma \left( r_{0},r_{1},...,r_{n-1}\right)
=\left( r_{n-1},r_{0},...,r_{n-2}\right) .$ If \ we apply $\phi ,$ we have 
\begin{eqnarray*}
\phi \left( \sigma \left( r_{0},...,r_{n-1}\right) \right) &=&\phi
(r_{n-1},r_{0},...,r_{n-2}) \\
&=&(a_{n-1},...,a_{n-2},a_{n-1}+b_{n-1}+c_{n-1},...,a_{n-2}+b_{n-2}+ \\
&&c_{n-2},a_{n-1}+2b_{n-1}+c_{n-1},...,a_{n-2}+2b_{n-2}+c_{n-2})
\end{eqnarray*}%
On the other hand $\phi
(r_{0},...,r_{n-1})=(a_{0},...,a_{n-1},a_{0}+b_{0}+c_{0},...,a_{n-1}+b_{n-1}+c_{n-1},a_{0}+2b_{0}+c_{0},...,a_{n-1}+2b_{n-1}+c_{n-1}) 
$. If we apply $\varphi ,$ we have $\varphi (\phi (r_{0},r_{1},...,$\newline
$%
r_{n-1}))=(a_{n-1},...,a_{n-2},a_{n-1}+b_{n-1}+c_{n-1},...,a_{n-2}+b_{n-2}+c_{n-2},a_{n-1}+2b_{n-1}+c_{n-1},...,a_{n-2}+2b_{n-2}+c_{n-2}) 
$. Thus, $\phi \sigma =\varphi \phi .$
\end{proof}

\begin{proposition}
Let $\sigma $ and $\varphi $ be as in the preliminaries. A code $C$ of
length $n$ over $R$ is cyclic code if and only if $\phi \left( C\right) $ is
quasi cyclic code of index $3$ over $Z_{3}$ with length $3n$.
\end{proposition}

\begin{proof}
Suppose $C$ is cyclic code. Then $\sigma \left( C\right) =C.$ If we apply $%
\phi ,$ we have $\phi \left( \sigma \left( C\right) \right) =\phi \left(
C\right) .$ From proposition 4, $\phi \left( \sigma \left( C\right) \right)
=\varphi \left( \phi \left( C\right) \right) =\phi \left( C\right) .$ Hence, 
$\phi \left( C\right) $ is a quasi cyclic code of index $3$. Conversely, if $%
\phi \left( C\right) $ is a quasi cyclic code of index $3$, then $\varphi
(\phi \left( C\right) )=\phi \left( C\right) .$ From proposition 4, we have $%
\varphi \left( \phi \left( C\right) \right) =\phi \left( \sigma \left(
C\right) \right) =\phi \left( C\right) .$ Since $\phi $ is injective, it
follows that $\sigma \left( C\right) =C.$
\end{proof}

\begin{proposition}
Let $\tau _{s,l}$ be quasi-cyclic shift on $R$. Let $\Gamma $ be as in the
preliminaries. Then $\phi \tau _{s,l}=\Gamma \phi .$
\end{proposition}

\begin{proof}
Let $e=\left(
e_{0,0},...,e_{0,l-1},e_{1,0},...,e_{1,l-1},...,e_{s-1,0},...,e_{s-1,l-1}%
\right) $ with $e_{i,j}=a_{i,j}+vb_{i,j}+v^{2}c_{i,j}$ where $i=0,1,...,s-1$
and $j=0,1,...,l-1$. We have $\tau _{s,l}\left( e\right) =\left(
e_{s-1,0},...,e_{s-1,l-1},e_{0,0},...,e_{0,l-1},...,e_{s-2,0},...,e_{s-2,l-1}\right) 
$. If we apply $\phi $, we have%
\begin{eqnarray*}
\phi (\tau _{s,l}\left( e\right) )
&=&(a_{s-1,0},...,a_{s-2,l-1},a_{s-1,0}+b_{s-1,0}+c_{s-1,0},...,a_{s-2,l-1}+
\\
&&b_{s-2,l-1}+c_{s-2,l-1},a_{s-1,0}+2b_{s-1,0}+c_{s-1,0},...,a_{s-2,l-1}+ \\
&&2b_{s-2,l-1}+c_{s-2,l-1})
\end{eqnarray*}%
On the other hand,%
\begin{eqnarray*}
\phi (e)
&=&(a_{0,0},...,a_{s-1,l-1},a_{0,0}+b_{0,0}+c_{0,0},...,a_{s-1,l-1}+b_{s-1,l-1}+c_{s-1,l-1},
\\
&&a_{0,0}+2b_{0,0}+c_{0,0},...,a_{s-1,l-1}+2b_{s-1,l-1}+c_{s-1,l-1})
\end{eqnarray*}%
$\Gamma (\varphi
(e))=(a_{s-1,0},...,a_{s-2,l-1},a_{s-1,0}+b_{s-1,0}+c_{s-1,0},...,a_{s-2,l-1}+b_{s-2,l-1}+c_{s-2,l-1},a_{s-1,0}+2b_{s-1,0}+c_{s-1,0},...,a_{s-2,l-1}+2b_{s-2,l-1}+c_{s-2,l-1}) 
$. So, we have $\varphi \tau _{s,l}=\Gamma \varphi $.
\end{proof}

\begin{theorem}
The Gray image of a quasi-cyclic code over $R$ of length $n$\ with index $l$
is a $l$-quasi cyclic code of index $3$ over $Z_{3}$ with length $3n$.
\end{theorem}

\begin{proof}
Let $C$ be a quasi-cyclic code over $R$ of length $n$ with index $l$. That
is $\tau _{s,l}\left( C\right) =C$. If we apply $\phi $, we have $\phi (\tau
_{s,l}\left( C\right) )=\phi (C)$. From the Proposition $6$, $\phi \left(
\tau _{s,l}\left( C\right) \right) =\phi \left( C\right) =\Gamma \left( \phi
\left( C\right) \right) $. So, $\phi \left( C\right) $ is a $l$ quasi-cyclic
code of index $3$ over $Z_{3}$ with length $3n$.
\end{proof}

We denote that $A_{1}\otimes A_{2}\otimes A_{3}\otimes
A_{4}=\{(a_{1},a_{2},a_{3},a_{4}):a_{1}\in A_{1},a_{2}\in A_{2},a_{3}\in
A_{3},a_{4}\in A_{4}\}$ and $A_{1}\oplus A_{2}\oplus A_{3}\oplus
A_{4}=\{a_{1}+a_{2}+a_{3}+a_{4}:a_{1}\in A_{1},a_{2}\in A_{2},a_{3}\in
A_{3},a_{4}\in A_{4}\}$

Let $C$ be a linear code of length $n$ over $R.$ Define%
\begin{eqnarray*}
C_{1} &=&\left\{ a\in Z_{3}^{n}:\exists b,c\in Z_{3}^{n},a+vb+v^{2}c\in
C\right\} \\
C_{2} &=&\left\{ a+b+c\in Z_{3}^{n}:a+vb+v^{2}c\in C\right\} \\
C_{3} &=&\left\{ a+2b+c\in Z_{3}^{n}:a+vb+v^{2}c\in C\right\}
\end{eqnarray*}

Then $C_{1},C_{2}$ and $C_{3}$ are ternary linear codes of length $n$.
Moreover, the linear code $C$ of length $n$ over $R$ can be uniquely
expressed as $C=(1+2v^{2})C_{1}\oplus \left( 2v+2v^{2}\right) C_{2}\oplus
\left( v+2v^{2}\right) C_{3}.$

\begin{theorem}
Let $C$ be a linear code of length $n$ over $R.$ Then $\phi \left( C\right)
=C_{1}\otimes C_{2}\otimes C_{3}$ and $\left\vert C\right\vert =\left\vert
C_{1}\right\vert \left\vert C_{2}\right\vert \left\vert C_{3}\right\vert .$
\end{theorem}

\begin{proof}
For any $%
(a_{0},a_{1},...,a_{n-1},a_{0}+b_{0}+c_{0},a_{1}+b_{1}+c_{1},...,a_{n-1}+b_{n-1}+c_{n-1},a_{0}+2b_{0}+c_{0},a_{1}+2b_{1}+c_{1},...,a_{n-1}+2b_{n-1}+c_{n-1})\in \phi \left( C\right) . 
$ Let $m_{i}=a_{i}+vb_{i}+v^{2}c_{i},$ $i=0,1,...,n-1.$ Since\ $\phi $ is a
bijection $m=\left( m_{0},m_{1},...,m_{n-1}\right) \in C.$ By definitions of 
$C_{1},C_{2}$ and $C_{3}$ we have $\left( a_{0},a_{1},...,a_{n-1}\right) \in
C_{1},(a_{0}+b_{0}+c_{0},a_{1}+b_{1}$\newline
$+c_{1},...,a_{n-1}+b_{n-1}+c_{n-1})\in
C_{2},(a_{0}+2b_{0}+c_{0},a_{1}+2b_{1}+c_{1},...,a_{n-1}+2b_{n-1}+c_{n-1})%
\in C_{3}$. So, $%
(a_{0},a_{1},...,a_{n-1},a_{0}+b_{0}+c_{0},a_{1}+b_{1}+c_{1},...,a_{n-1}+b_{n-1}+c_{n-1},a_{0}+2b_{0}+c_{0},a_{1}+2b_{1}+c_{1},...,a_{n-1}+2b_{n-1}+c_{n-1})\in C_{1}\otimes C_{2}\otimes C_{3} 
$. That is $\phi \left( C\right) \subseteq C_{1}\otimes C_{2}\otimes C_{3}$.

On the other hand, for any $(a,b,c)\in C_{1}\otimes C_{2}\otimes C_{3}$
where $a=(a_{0},a_{1},...,a_{n-1})\in C_{1},$ $%
b=(a_{0}+b_{0}+c_{0},a_{1}+b_{1}+c_{1},...,a_{n-1}+b_{n-1}+c_{n-1})\in C_{2}$%
, $c=(a_{0}+2b_{0}+c_{0},a_{1}+2b_{1}+c_{1},...,a_{n-1}+2b_{n-1}+c_{n-1})\in
C_{3}$. There are $x=\left( x_{0},x_{1},...,x_{n-1}\right) $, $y=\left(
y_{0},y_{1},...,y_{n-1}\right) $, $z=\left( z_{0},z_{1},...,z_{n-1}\right)
\in C$ such that $x_{i}=a_{i}+(v+2v^{2})p_{i},$ $y_{i}=b_{i}+\left(
1+2v^{2}\right) q_{i},$ $z_{i}=c_{i}+\left( 2v+2v^{2}\right) r_{i}$ where $%
p_{i},$ $q_{i},r_{i}\in Z_{3}$ and $0\leq i\leq n-1.$ Since $C$ is linear we
have $m=(1+2v^{2})x+\left( 2v+2v^{2}\right) y+\left( v+2v^{2}\right)
z=a+v(2b+c)+v^{2}(2a+2b+2c)\in C.$ It follows then $\phi \left( m\right)
=\left( a,b,c\right) $, which gives $C_{1}\otimes C_{2}\otimes
C_{3}\subseteq \phi \left( C\right) .$

Therefore, $\phi \left( C\right) =C_{1}\otimes C_{2}\otimes C_{3}.$ The
second result is easy to verify.
\end{proof}

\begin{corollary}
If $\phi \left( C\right) =C_{1}\otimes C_{2}\otimes C_{3},$ then $%
C=(1+2v^{2})C_{1}\oplus \left( 2v+2v^{2}\right) C_{2}\oplus \left(
v+2v^{2}\right) C_{3}.$ It is easy to see that
\end{corollary}

$\left\vert C\right\vert =\left\vert C_{1}\right\vert \left\vert
C_{2}\right\vert \left\vert C_{3}\right\vert =3^{n-\deg (f_{1})}3^{n-\deg
(f_{2})}3^{n-\deg (f_{3})}$

$\qquad =3^{3n-(\deg (f_{1})+\deg (f_{2})+\deg (f_{3}))}$ where $f_{1},f_{2}$
and $f_{3}$ are the generator polynomials of $C_{1},C_{2}$ and $C_{3}$,
respectively.

\begin{corollary}
If $G_{1},G_{2}$ and $G_{3}$ are generator matrices of ternary linear codes $%
C_{1},C_{2}$ and $C_{3}$ respectively, then the generator matrix of $C$ is
\end{corollary}

$\ \ \ \ \ \ \ \ \ \ \ \ \ \ \ \ \ \ \ \ \ \ \ \ \ \ \ \ \ \ \ \ \ \ \ \ \ \
\ \ \ \ G=\left[ 
\begin{array}{c}
(1+2v^{2})G_{1} \\ 
\left( 2v+2v^{2}\right) G_{2} \\ 
\left( v+2v^{2}\right) G_{3}%
\end{array}%
\right] $ .

We have$\ $

$\ $%
\begin{equation*}
\phi (G)=\left[ 
\begin{array}{c}
\phi ((1+2v^{2})G_{1}) \\ 
\phi (\left( 2v+2v^{2}\right) G_{2}) \\ 
\phi (\left( v+2v^{2}\right) G_{3})%
\end{array}%
\right] =\left[ 
\begin{array}{ccc}
G_{1} & 0 & 0 \\ 
0 & G_{2} & 0 \\ 
0 & 0 & G_{3}%
\end{array}%
\right] .
\end{equation*}%
$\ \ \ \ \ \ \ \ \ \ \ \ \ \ \ \ \ \ \ \ \ \ \ \ \ \ \ \ \ \ \ \ \ \ \ \ \ \
\ \ \ \ \ \ \ \ $

\bigskip

Let $d_{L}$ minimum Lee weight of linear code $C$ over $R$. Then,

\qquad \qquad \qquad\ $d_{L}=d_{H}(\phi \left( C\right) )=\min
\{d_{H}(C_{1}),d_{H}(C_{2}),d_{H}(C_{3})\}$

where $d_{H}(C_{i})$ denotes the minimum Hamming weights of ternary codes $%
C_{1},C_{2}$ and $C_{3}$, respectively.

\bigskip As similiar to section 4 in [1] we have the following Lemma and
Examples.

\begin{lemma}
Let $C=\left\langle f(x)\right\rangle $ be a negacyclic code of length n
over R and $\phi \left( f(x)\right) =\left( f_{1},f_{2},f_{3}\right) $ with $%
\deg (\gcd (f_{1},x^{n}+1))=n-k_{1},\deg (\gcd (f_{2},x^{n}+1))=n-k_{2},\deg
(\gcd (f_{3},x^{n}+1))=n-k_{3}.$ Then, $\left\vert C\right\vert
=3^{k_{1}+k_{2}+k_{3}}.$
\end{lemma}

\begin{example}
Let $C=\left\langle f(x)\right\rangle =\left\langle
(2v+2v^{2})x^{2}+(1+2v+2v^{2})x+1\right\rangle $ be a negacyclic code of
length $3$ over $R$. Hence, $\phi \left( f(x)\right) =(x+1,x^{2}+2x+1,x+1)$
and 
\begin{eqnarray*}
f_{1} &=&\gcd (x+1,x^{3}+1)=x+1 \\
f_{2} &=&\gcd (x^{2}+2x+1,x^{3}+1)=x^{2}+2x+1 \\
f_{3} &=&\gcd (x+1,x^{3}+1)=x+1
\end{eqnarray*}%
So we have $\left\vert C\right\vert =3^{2+1+2}=3^{5}.$
\end{example}

\begin{example}
Let $C=\left\langle f(x)\right\rangle =\left\langle
v^{2}x^{4}+vx^{3}+(1+2v^{2})x^{2}+2vx+1\right\rangle $ be a negacyclic code
of length $10$ over $R$. Hence, $\phi \left( f(x)\right)
=(x^{2}+1,x^{4}+x^{3}+2x+1,x^{4}+2x^{3}+x+1)$ and 
\begin{eqnarray*}
f_{1} &=&\gcd (x^{2}+1,x^{10}+1)=x^{2}+1 \\
f_{2} &=&\gcd (x^{4}+x^{3}+2x+1,x^{10}+1)=x^{4}+x^{3}+2x+1 \\
f_{3} &=&\gcd (x^{4}+2x^{3}+x+1,x^{10}+1)=x^{4}+2x^{3}+x+1
\end{eqnarray*}%
So we have $\left\vert C\right\vert =3^{8+6+6}=3^{20}.$
\end{example}

Let $h_{i}(x)=(x^{n}+1)/(\gcd (x^{n}+1,f_{i})).$ Hence, $C^{\bot
}=\left\langle \phi ^{-1}\left(
h_{1_{R}}(x),h_{2_{R}}(x),h_{3_{R}}(x)\right) \right\rangle $ where $%
h_{i_{R}}(x)$ be the reciprocal polynomial of $h_{i}(x)$ for $i=1,2,3.$By
using the previous example, 
\begin{eqnarray*}
C^{\bot } &=&\left\langle \phi ^{-1}\left(
h_{1_{R}}(x),h_{2_{R}}(x),h_{3_{R}}(x)\right) \right\rangle  \\
&=&\left\langle \phi
^{-1}(x^{8}+2x^{6}+x^{4}+2x^{2}+1,x^{6}+x^{5}+x^{4}+x^{2}+2x+1,x^{6}+2x^{5}+x^{4}+x^{2}+x+1)\right\rangle 
\\
&=&\left\langle
(1+2v^{2})x^{8}+(2+2v^{2})x^{6}+vx^{5}+x^{4}+(2+2v^{2})x^{2}+2vx+1\right%
\rangle 
\end{eqnarray*}

\section{Quantum Codes From Cyclic Codes Over $R$}

\begin{theorem}
Let $C_{1}=\left[ n,k_{1},d_{1}\right] _{q}$ and $C_{2}=\left[ n,k_{2},d_{2}%
\right] _{q}$ be linear codes over GF(q) with $C_{2}^{\perp }\subseteq C_{1}.
$ Furthermore, let $d=min\{wt(v):v\in (C_{1}\backslash C_{2}^{\perp })\cup
(C_{2}^{\perp }\backslash C_{1})\}\geq \min \{d_{1},d_{2}\}.$Then there
exists a quantum error-correcting code $C=[n,k_{1}+k_{2}-n,d]_{q}.$In
particular, if $C_{1}^{\perp }\subseteq C_{1}$, then there exists a quantum
error-correcting code $C=[n,n-2k_{1},d_{1}],$ where $d_{1}=min\{wt(v):v\in
(C_{1}^{\perp }\backslash C_{1})\}$, $[20]$.
\end{theorem}

\begin{proposition}
Let $C=(1+2v^{2})C_{1}\oplus \left( 2v+2v^{2}\right) C_{2}\oplus \left(
v+2v^{2}\right) C_{3}$ be a linear code over $R.$Then $C$ is a cyclic code
over $R$ iff $C_{1},C_{2}$ and $C_{3}$ are cyclic codes.
\end{proposition}

\begin{proof}
Let $\left( a_{0},a_{1},...,a_{n-1}\right) \in C_{1},$ $\left(
b_{0},b_{1},...,b_{n-1}\right) \in C_{2}$ and $\left(
c_{0},c_{1},...,c_{n-1}\right) \in C_{3}$ . Assume that $%
m_{i}=(1+2v^{2})a_{i}+\left( 2v+2v^{2}\right) b_{i}+\left( v+2v^{2}\right)
c_{i}$ for $i=0,1,...,n-1$. Then $\left( m_{0},m_{1},...,m_{n-1}\right) \in C
$. Since $C$ is a cyclic code, it follows that $\left(
m_{n-1},m_{0},...,m_{n-2}\right) \in C$. Note that $%
(m_{n-1},m_{0},...,m_{n-2})=(1+2v^{2})(a_{n-1},a_{0},...,a_{n-2})+\left(
2v+2v^{2}\right) (b_{n-1},b_{0},...,$ $b_{n-2})+\left( v+2v^{2}\right)
(c_{n-1},c_{0},...,c_{n-2})$. Hence $(a_{n-1},a_{0},...,a_{n-2})\in C_{1},$ $%
\left( b_{n-1},b_{0},...,b_{n-2}\right) \in C_{2}$ and $%
(c_{n-1},c_{0},...,c_{n-2})\in C_{3}$. Therefore,$C_{1},C_{2}$ and $C_{3}$
cyclic codes over $Z_{3}.$

Conversely, suppose that $C_{1},C_{2}$ and $C_{3}$ cyclic codes over $Z_{3}$%
. Let $(m_{0},m_{1},...,$\newline
$m_{n-1})\in C$ where $m_{i}=(1+2v^{2})a_{i}+\left( 2v+2v^{2}\right)
b_{i}+\left( v+2v^{2}\right) c_{i}$ for $i=0,1,...,n-1$. Then $\left(
a_{0},a_{1},...,a_{n-1}\right) \in C_{1},$ $\left(
b_{0},b_{1},...,b_{n-1}\right) \in C_{2}$ and $(c_{0},c_{1},...,$\newline
$c_{n-1})\in C_{3}$. Note that $\left( m_{n-1},m_{0},...,m_{n-2}\right)
=(1+2v^{2})(a_{n-1},a_{0},...,a_{n-2})+\left( 2v+2v^{2}\right)
(b_{n-1},b_{0},...,b_{n-2})+\left( v+2v^{2}\right)
(c_{n-1},c_{0},...,c_{n-2})\in C=(1+2v^{2})C_{1}\oplus \left(
2v+2v^{2}\right) C_{2}\oplus \left( v+2v^{2}\right) C_{3}.$ So, $C$ is
cyclic code over $R.$
\end{proof}

\begin{proposition}
Let $C=(1+2v^{2})C_{1}\oplus \left( 2v+2v^{2}\right) C_{2}\oplus \left(
v+2v^{2}\right) C_{3}$ be a linear code over $R.$Then $C$ is a nega-cyclic
code over $R$ iff $C_{1},C_{2}$ and $C_{3}$ are nega-cyclic codes.
\end{proposition}

\begin{proof}
Let $\left( a_{0},a_{1},...,a_{n-1}\right) \in C_{1},$ $\left(
b_{0},b_{1},...,b_{n-1}\right) \in C_{2}$ and $\left(
c_{0},c_{1},...,c_{n-1}\right) \in C_{3}$ . Assume that $%
m_{i}=(1+2v^{2})a_{i}+\left( 2v+2v^{2}\right) b_{i}+\left( v+2v^{2}\right)
c_{i}$ for $i=0,1,...,n-1$. Then $\left( m_{0},m_{1},...,m_{n-1}\right) \in C
$. Since $C$ is a nega-cyclic code, it follows that $\left(
-m_{n-1},m_{0},...,m_{n-2}\right) \in C$. Note that $%
(-m_{n-1},m_{0},...,m_{n-2})=(1+2v^{2})(-a_{n-1},a_{0},...,a_{n-2})+\left(
2v+2v^{2}\right) (-b_{n-1},b_{0},...,$ $b_{n-2})+\left( v+2v^{2}\right)
(-c_{n-1},c_{0},$\newline
$...,c_{n-2})$. Hence $(-a_{n-1},a_{0},...,a_{n-2})\in C_{1},\left(
-b_{n-1},b_{0},...,b_{n-2}\right) \in C_{2}$ and $(-c_{n-1}$\newline
$,c_{0},...,c_{n-2})\in C_{3}$. Therefore, $C_{1},C_{2}$ and $C_{3}$
nega-cyclic codes over $Z_{3}.$

Conversely, suppose that $C_{1},C_{2}$ and $C_{3}$ nega-cyclic codes over $%
Z_{3}$. Let $(m_{0},m_{1},...,m_{n-1})\in C$ where $m_{i}=(1+2v^{2})a_{i}+%
\left( 2v+2v^{2}\right) b_{i}+\left( v+2v^{2}\right) c_{i}$ for $%
i=0,1,...,n-1$. Then $\left( a_{0},a_{1},...,a_{n-1}\right) \in C_{1},$ $%
\left( b_{0},b_{1},...,b_{n-1}\right) \in C_{2}$ and $\left(
c_{0},c_{1},...,c_{n-1}\right) \in C_{3}$ . Note that $\left(
-m_{n-1},m_{0},...,m_{n-2}\right) =(1+2v^{2})(-a_{n-1},a_{0},.$\newline
$..,a_{n-2})+\left( 2v+2v^{2}\right) (-b_{n-1},b_{0},...,b_{n-2})+\left(
v+2v^{2}\right) (-c_{n-1},c_{0},...,c_{n-2})\in C=(1+2v^{2})C_{1}\oplus
\left( 2v+2v^{2}\right) C_{2}\oplus \left( v+2v^{2}\right) C_{3}.$ So, $C$
is nega-cyclic code over $R.$
\end{proof}

\begin{proposition}
Suppose $C=(1+2v^{2})C_{1}\oplus \left( 2v+2v^{2}\right) C_{2}\oplus \left(
v+2v^{2}\right) C_{3}$ is a cyclic (negacyclic) code of length $n$ over $R.$%
Then 
\begin{equation*}
C=<(1+2v^{2})f_{1},\left( 2v+2v^{2}\right) f_{2},\left( v+2v^{2}\right)
f_{3}>
\end{equation*}%
\newline
and $\left\vert C\right\vert =3^{3n-(\deg f_{1}+\deg f_{2}+\deg f_{3})}$
where $f_{1},f_{2}$ and $f_{3}$ generator polynomials of $C_{1},C_{2}$ and $%
C_{3}$ respectively.
\end{proposition}

\begin{proposition}
Suppose $C$ is a cyclic (negacyclic) code of length $n$ over $R$, then there
is a unique polynomial $f\left( x\right) $ such that $C=\left\langle f\left(
x\right) \right\rangle $ and $f\left( x\right) \mid x^{n}-1$ ($f\left(
x\right) \mid x^{n}+1$) where $f\left( x\right) =(1+2v^{2})f_{1}(x)+\left(
2v+2v^{2}\right) f_{2}(x)+\left( v+2v^{2}\right) f_{3}(x).$
\end{proposition}

\begin{proposition}
Let $C$ be a linear code of length $n$ over $R$, then $C^{\bot
}=(1+2v^{2})C_{1}^{\bot }\oplus \left( 2v+2v^{2}\right) C_{2}^{\bot }\oplus
\left( v+2v^{2}\right) C_{3}^{\bot }.$ Furthermore, $C$ is self-dual code
iff $C_{1},C_{2}$ and $C_{3}$ are self-dual codes over $Z_{3}.$
\end{proposition}

\begin{proposition}
If \ $C=(1+2v^{2})C_{1}\oplus \left( 2v+2v^{2}\right) C_{2}\oplus \left(
v+2v^{2}\right) C_{3}$ is a cyclic (negacyclic) code of length $n$ over $R$.
Then 
\begin{equation*}
C^{\perp }=\left\langle (1+2v^{2})h_{1}^{\ast }+\left( 2v+2v^{2}\right)
h_{2}^{\ast }+\left( v+2v^{2}\right) h_{3}^{\ast }\right\rangle 
\end{equation*}%
and $\left\vert C^{\perp }\right\vert =3^{\deg f_{1}+\deg f_{2}+\deg f_{3}}$
where for $i=1,2,3$, $h_{i}^{\ast }$ are the reciprocal polynomials of $h_{i}
$ i.e., $h_{i}\left( x\right) =\left( x^{n}-1\right) /f_{i}\left( x\right) ,$%
($h_{i}\left( x\right) =\left( x^{n}+1\right) /f_{i}\left( x\right) $), $%
h_{i}^{\ast }\left( x\right) =x^{\deg h_{i}}h_{i}\left( x^{-1}\right) $ for $%
i=1,2,3$.
\end{proposition}

\begin{lemma}
A ternary linear cyclic (negacyclic) code $C$ with generator polynomial $%
f\left( x\right) $ contains its dual code iff
\end{lemma}

\begin{equation*}
x^{n}-1\equiv 0\left( \func{mod}ff^{\ast }\right) ,\text{ \ \ \ \ \ (}%
x^{n}+1\equiv 0\left( \func{mod}ff^{\ast }\right) \text{)}
\end{equation*}

where $f^{\ast }$ is the reciprocal polynomial of $f$.

\begin{theorem}
Let $C=\left\langle (1+2v^{2})f_{1},\left( 2v+2v^{2}\right) f_{2},\left(
v+2v^{2}\right) f_{3}\right\rangle $ be a cyclic (negacyclic ) code of
length $n$ over $R$. Then $C^{\perp }\subseteq C$ iff $x^{n}-1\equiv 0\left( 
\func{mod}f_{i}f_{i}^{\ast }\right) $ ($x^{n}+1\equiv 0\left( \func{mod}%
f_{i}f_{i}^{\ast }\right) $) for $i=1,2,3.$
\end{theorem}

\begin{proof}
Let $x^{n}-1\equiv 0\left( \func{mod}f_{i}f_{i}^{\ast }\right) $ ($%
x^{n}+1\equiv 0\left( \func{mod}f_{i}f_{i}^{\ast }\right) $) for $i=1,2,3.$
Then $C_{1}^{\perp }\subseteq C_{1},C_{2}^{\perp }\subseteq
C_{2},C_{3}^{\perp }\subseteq C_{3}.$ By using $(1+2v^{2})C_{1}^{\perp
}\subseteq (1+2v^{2})C_{1},$ $\left( 2v+2v^{2}\right) C_{2}^{\perp
}\subseteq \left( 2v+2v^{2}\right) C_{2},$ $\left( v+2v^{2}\right)
C_{3}^{\perp }\subseteq \left( v+2v^{2}\right) C_{3}$ $.$ We have $%
(1+2v^{2})C_{1}^{\perp }\oplus \left( 2v+2v^{2}\right) C_{2}^{\perp }\oplus
\left( v+2v^{2}\right) C_{3}^{\perp }$\newline
$\subseteq (1+2v^{2})C_{1}\oplus \left( 2v+2v^{2}\right) C_{2}\oplus \left(
v+2v^{2}\right) C_{3}.$ So, $<(1+2v^{2})h_{1}^{\ast }+\left(
2v+2v^{2}\right) h_{2}^{\ast }+\left( v+2v^{2}\right) h_{3}^{\ast
}>\subseteq $ $<(1+2v^{2})f_{1},\left( 2v+2v^{2}\right) f_{2},\left(
v+2v^{2}\right) f_{3}>.$ That is $C^{\perp }\subseteq C$.

Conversely, if $C^{\perp }\subseteq C$, then $(1+2v^{2})C_{1}^{\perp }\oplus
\left( 2v+2v^{2}\right) C_{2}^{\perp }\oplus \left( v+2v^{2}\right)
C_{3}^{\perp }\subseteq (1+2v^{2})C_{1}\oplus \left( 2v+2v^{2}\right)
C_{2}\oplus \left( v+2v^{2}\right) C_{3}.$ By thinking $\func{mod}(1+2v^{2}),%
\func{mod}\left( 2v+2v^{2}\right) $ and $\func{mod}\left( v+2v^{2}\right) $
respectively we have $C_{i}^{\perp }\subseteq C_{i}$ for $i=1,2,3$.
Therefore, $x^{n}-1\equiv 0\left( \func{mod}f_{i}f_{i}^{\ast }\right) $ ($%
x^{n}+1\equiv 0\left( \func{mod}f_{i}f_{i}^{\ast }\right) $) for $i=1,2,3.$
\end{proof}

\begin{corollary}
$C=(1+2v^{2})C_{1}\oplus \left( 2v+2v^{2}\right) C_{2}\oplus \left(
v+2v^{2}\right) C_{3}$ is a cyclic (negacyclic) code of length $n$ over $R$.
Then $C^{\perp }\subseteq C$ iff $C_{i}^{\perp }\subseteq C_{i}$ for $%
i=1,2,3 $.
\end{corollary}

\begin{example}
Let $\ n=6,R=Z_{3}+vZ_{3}+v^{2}Z_{3},v^{3}=v.$ We have $%
x^{6}-1=(2x^{2}+2)(x^{2}+2)(2x^{2}+1)=f_{1}f_{2}f_{3}$ in $Z_{3}\left[ x%
\right] $. Hence,%
\begin{eqnarray*}
f_{1}^{\ast } &=&2x^{2}+2=f_{1} \\
f_{2}^{\ast } &=&2x^{2}+1=f_{3} \\
f_{3}^{\ast } &=&x^{2}+2=f_{2}
\end{eqnarray*}
\end{example}

Let $C=\left\langle (1+2v^{2})f_{2},\left( 2v+2v^{2}\right) f_{2},\left(
v+2v^{2}\right) f_{3}\right\rangle .$ Obviously $x^{6}-1$ is divisibly by $%
f_{i}f_{i}^{\ast }$ for $i=2,3$. Thus we have $C^{\perp }\subseteq C.$

\begin{example}
Let $\ n=10,R=Z_{3}+vZ_{3}+v^{2}Z_{3},v^{3}=v.$ We have $x^{10}+1=\left(
x^{2}+1\right) \left( x^{4}+x^{3}+2x+1\right) \left( x^{4}+2x^{3}+x+1\right)
=g_{1}g_{2}g_{3}$ in $Z_{3}\left[ x\right] $. Hence,%
\begin{eqnarray*}
g_{1}^{\ast } &=&x^{2}+1=g_{1} \\
g_{2}^{\ast } &=&x^{4}+2x^{3}+x+1=g_{3} \\
g_{3}^{\ast } &=&x^{4}+x^{3}+2x+1=g_{2}
\end{eqnarray*}
\end{example}

Let $C=\left\langle (1+2v^{2})g_{2},\left( 2v+2v^{2}\right) g_{2},\left(
v+2v^{2}\right) g_{3}\right\rangle .$ Obviously $x^{10}+1$ is divisibly by $%
g_{i}g_{i}^{\ast }$ for $i=2,3$. Thus we have $C^{\perp }\subseteq C.$

\begin{theorem}
Let $C$ be linear code of length $n$ over $R$ with $\left\vert C\right\vert
=3^{3k_{1}+2k_{2}+k_{3}}$ and minimum distance $d$. Then $\phi \left(
C\right) $ is ternary linear $\left[ 3n,3k_{1}+2k_{2}+k_{3},d\right] $ code.
\end{theorem}

Using Theorem $14$ and Theorem $22$ we can construct quantum codes.

\begin{theorem}
Let $(1+2v^{2})C_{1}\oplus \left( 2v+2v^{2}\right) C_{2}\oplus \left(
v+2v^{2}\right) C_{3}$ be a cyclic (negacyclic) code of arbitrary length $n$
over $R$ with type $27^{k_{1}}9^{k_{2}}3^{k_{3}}.$ If $C_{i}^{\perp
}\subseteq C_{i}$ where $i=1,2,3$ then $C^{\perp }\subseteq C$ and there
exists a quantum error-correcting code with parameters $\left[ \left[
3n,3k_{1}+2k_{2}+k_{3}-3n,d_{L}\right] \right] $ where $d_{L}$ is the
minimum Lee weights of $C.$
\end{theorem}

\begin{example}
Let $n=6.$ We have $x^{6}-1=(2x^{2}+2)(x^{2}+2)(2x^{2}+1)$ in $Z_{3}\left[ x%
\right] .$ Let $f_{1}\left( x\right) =f_{2}\left( x\right) =x^{2}+2$, $%
f_{3}=2x^{2}+1.$ Thus $C=<(1+2v^{2})f_{1},\left( 2v+2v^{2}\right) f_{2},$%
\newline
$\left( v+2v^{2}\right) f_{3}>$. $C$ is a linear cyclic code of length $6$.
\ The dual code $C^{\bot }=\left\langle (1+2v^{2})h_{1}^{\ast },\left(
2v+2v^{2}\right) h_{2}^{\ast },\left( v+2v^{2}\right) h_{3}^{\ast
}\right\rangle $ can be obtained of Proposition $20$.Clearly, $C^{\bot
}\subseteq C.$ Hence, we obtain a quantum code with parameters $\left[ \left[
18,6,2\right] \right] .$
\end{example}

\begin{example}
Let $n=8.$ We have $x^{8}-1=(x+1)(x+2)(x^{2}+1)(x^{2}+x+2)(x^{2}+2x+2)$ in $%
Z_{3}\left[ x\right] .$ Let $f_{1}\left( x\right) =f_{2}\left( x\right)
=f_{3}\left( x\right) =x^{2}+1.$ Thus $C=<(1+2v^{2})f_{1},\left(
2v+2v^{2}\right) f_{2},\left( v+2v^{2}\right) f_{3}>$. $C$ is a linear
cyclic code of length $8$. Hence, we obtain a quantum code with parameters $%
\left[ \left[ 24,12,2\right] \right] .$
\end{example}

\begin{example}
Let $n=12.$ We have $x^{12}-1=\left( x-1\right) ^{3}\left(
x^{3}+x^{2}+x+1\right) ^{3}$ in $Z_{3}\left[ x\right] .$ Let $f_{1}\left(
x\right) =f_{2}\left( x\right) =f_{3}\left( x\right) =x^{3}+x^{2}+x+1$.Thus $%
C=<(1+2v^{2})f_{1},\left( 2v+2v^{2}\right) f_{2},\left( v+2v^{2}\right)
f_{3}>$. $C$ is a linear cyclic code of length $12$. \ The dual code $%
C^{\bot }=\left\langle (1+2v^{2})h_{1}^{\ast },\left( 2v+2v^{2}\right)
h_{2}^{\ast },\left( v+2v^{2}\right) h_{3}^{\ast }\right\rangle $ can be
obtained of Proposition $20$.Clearly, $C^{\bot }\subseteq C.$ Hence, we
obtain a quantum code with parameters $\left[ \left[ 36,18,2\right] \right] .
$
\end{example}

Let $n=27.$ We have $x^{27}-1=\left( x-1\right)
^{3}(x^{3}-1)^{4}(x^{6}-2x^{3}+1)^{2}$ in $Z_{3}\left[ x\right] .$ Let $%
f_{1}\left( x\right) =f_{2}\left( x\right) =f_{3}(x)=x^{6}-2x^{3}+1.$ Hence,
we obtain a quantum code with parameters $\left[ \left[ 81,45,2\right] %
\right] .$

Let $n=30.$ We have $x^{30}-1=\left( x^{2}+2\right) ^{3}\left(
x^{4}+x^{3}+x^{2}+x+1\right) ^{3}(x^{4}+2x^{3}+$\newline
$x^{2}+2x+1)^{3}$ in $Z_{3}\left[ x\right] .$ Let $f_{1}\left( x\right)
=f_{3}\left( x\right) =x^{4}+x^{3}+x^{2}+x+1$, $f_{2}\left( x\right)
=x^{4}+2x^{3}+x^{2}+2x+1$. Hence, we obtain a quantum code with parameters $%
\left[ \left[ 90,66,2\right] \right] .$\newline

\begin{example}
Let $n=3.$ We have $x^{3}+1=\left( x+1\right) ^{3}$ in $Z_{3}\left[ x\right]
.$ Let $f_{1}\left( x\right) =f_{2}\left( x\right) =f_{3}\left( x\right)
=x+1.$ Thus $C=\left\langle (1+2v^{2})f_{1},\left( 2v+2v^{2}\right)
f_{2},\left( v+2v^{2}\right) f_{3}\right\rangle $. $C$ is a linear
negacyclic code of length $3$. \ The dual code $C^{\bot
}=<(1+2v^{2})h_{1}^{\ast },(2v+$\newline
$2v^{2})h_{2}^{\ast },\left( v+2v^{2}\right) h_{3}^{\ast }>$ can be obtained
of Proposition $20$.Clearly, $C^{\bot }\subseteq C.$ Hence, we obtain a
quantum code with parameters $\left[ \left[ 9,3,2\right] \right] .$
\end{example}

\begin{example}
Let $n=10.$ We have $x^{10}+1=\left( x^{2}+1\right) \left(
x^{4}+x^{3}+2x+1\right) (x^{4}+$\newline
$2x^{3}+x+1)$ in $Z_{3}\left[ x\right] .$ Let $f_{1}\left( x\right)
=x^{4}+x^{3}+2x+1,$ $f_{2}\left( x\right) =f_{3}\left( x\right)
=x^{4}+2x^{3}+x+1$. Thus $C=\left\langle (1+2v^{2})f_{1},\left(
2v+2v^{2}\right) f_{2},\left( v+2v^{2}\right) f_{3}\right\rangle $. $C$ is a
linear negacyclic code of length $10$. \ The dual code $C^{\bot
}=\left\langle (1+2v^{2})h_{1}^{\ast },\left( 2v+2v^{2}\right) h_{2}^{\ast
},\left( v+2v^{2}\right) h_{3}^{\ast }\right\rangle $ can be obtained of
Proposition $20$.Clearly, $C^{\bot }\subseteq C.$ Hence, we obtain a quantum
code with parameters $\left[ \left[ 30,6,4\right] \right] .$
\end{example}

\begin{example}
Let $n=12.$ We have $x^{12}+1=\left( x^{4}+1\right) \left( x^{2}+x+2\right)
\left( x^{2}+2x+2\right) $\newline
$\left( 2x^{2}+2x+1\right) \left( 2x^{2}+x+1\right) $ in $Z_{3}\left[ x%
\right] .$ Let $f_{1}\left( x\right) =x^{2}+x+2,$ $f_{2}\left( x\right)
=2x^{2}+x+1$, $f_{3}\left( x\right) =x^{2}+2x+2$. Thus $C=\left\langle
(1+2v^{2})f_{1},\left( 2v+2v^{2}\right) f_{2},\left( v+2v^{2}\right)
f_{3}\right\rangle $. $C$ is a linear negacyclic code of length $12$. \ The
dual code $C^{\bot }=<(1+2v^{2})h_{1}^{\ast },\left( 2v+2v^{2}\right)
h_{2}^{\ast },\left( v+2v^{2}\right) h_{3}^{\ast }>$ can be obtained of
Proposition $20$.Clearly, $C^{\bot }\subseteq C.$ Hence, we obtain a quantum
code with parameters $\left[ \left[ 36,24,2\right] \right] .$
\end{example}

\section{Constacyclic codes over $R$}

Let $\lambda =\alpha +\beta v+\gamma v^{2}$ be unit element of $R$. Note
that $\lambda ^{n}=1$ if $n$ even $\lambda ^{n}=\lambda $ if $n$ even. So we
only study $\lambda $-constacyclic codes of odd length.

\begin{proposition}
Let $\varrho $ be the map of $R\left[ x\right] /\left\langle
x^{n}-1\right\rangle $ into $R\left[ x\right] /\left\langle x^{n}-\lambda
\right\rangle $ defined by $\varrho (a(x))=a(\lambda x)$. If $n$ is odd,
then $\varrho $ is a ring isomorphism.
\end{proposition}

\begin{proof}
The proof is straightforward if $n$ is odd,

$\qquad a(x)\equiv b(x)(\func{mod}(x^{n}-1))$ iff $a(\lambda x)\equiv
b(\lambda x)(\func{mod}(x^{n}-\lambda ))$
\end{proof}

\begin{corollary}
I is an ideal of $R\left[ x\right] /\left\langle x^{n}-1\right\rangle $ if
and only if $\varrho (I)$ is an ideal of $R\left[ x\right] /\left\langle
x^{n}-\lambda \right\rangle .$
\end{corollary}

\begin{corollary}
Let $\overline{\varrho }$ be the permutation of $R^{n}$ with n odd, such
that $\overline{\varrho }(a_{0},a_{1},...,$\newline
$a_{n-1})=(a_{0},\lambda a_{1},\lambda ^{2}a_{2}...,\lambda ^{n-1}a_{n-1})$
and $\digamma $ be a subset of $R^{n}$ then $\digamma $ is a linear cyclic
code iff $\overline{\varrho }(\digamma )$ is a linear $\lambda $%
-constacyclic code.
\end{corollary}

\begin{corollary}
$C$ is a cyclic code of parameters $(n,3^{k},d)$ over $R$ iff $\overline{%
\varrho }(C)$ is a $\lambda $-constacyclic code of parameters $(n,3^{k},d)$
over $R,$ when $n$ is odd.
\end{corollary}

\begin{theorem}
Let $\lambda $ be a unit in $R$. Let $C=(1+2v^{2})C_{1}\oplus \left(
2v+2v^{2}\right) C_{2}\oplus \left( v+2v^{2}\right) C_{3}$ be a linear code
of length $n$ over $R$. Then $C$ is a $\lambda $-constacyclic code of length 
$n$ over $R$ iff $C_{i}$ is either a cyclic code or a negacyclic code of
length $n$ over $Z_{3}$ for $i=1,2,3$.
\end{theorem}

\begin{proof}
Let $\nu $ be the $\lambda $-constacyclic shift on $R^{n}.$ Let $C$ be a $%
\lambda $-constacyclic code of length $n$ over $R.$Let $\left(
a_{0},a_{1},...,a_{n-1}\right) \in C_{1},$ $\left(
b_{0},b_{1},...,b_{n-1}\right) \in C_{2}$ and $\left(
c_{0},c_{1},...,c_{n-1}\right) \in C_{3}.$ Then the corresponding element of 
$C$ is $(m_{0},m_{1},...,m_{n-1})=(1+2v^{2})(a_{0},a_{1},...,a_{n-1})+\left(
2v+2v^{2}\right) (b_{0},b_{1},...,$ $b_{n-1})+\left( v+2v^{2}\right)
(c_{0},c_{1},...,c_{n-1}).$ Since $C$ is a $\lambda $-constacyclic code so, $%
\nu \left( m\right) =(\lambda m_{n-1},m_{0},...,m_{n-2})\in C$ where $%
m_{i}=a_{i}+b_{i}v+v^{2}c_{i}$ for $i=0,1,...,n-1.$ Let $\lambda =\alpha
+v\beta +v^{2}\gamma ,$ where $\alpha ,\beta ,\gamma \in Z_{3}.$ $\nu \left(
m\right) =(1+2v^{2})(\lambda a_{n-1},a_{0},...,a_{n-2})+\left(
2v+2v^{2}\right) (\lambda b_{n-1},b_{0},...,$ $b_{n-2})+\left(
v+2v^{2}\right) (\lambda c_{n-1},c_{0},...,c_{n-2}).$ Since the units of $%
Z_{3}$ are $1$ and $-1,$ so $\alpha =\overline{+}1$. Therefore we have
obtained the desired result. The other side it is seen easily.
\end{proof}

\section{Skew Codes Over $R$}

We are interested in studying skew codes using the ring $%
R=Z_{3}+vZ_{3}+v^{2}Z_{3}$ where $v^{3}=v$. We define non-trivial ring
automorphism $\theta $ on the ring $R$ by $\theta \left( a+vb+v^{2}c\right)
=a+2bv+v^{2}c$ for all $a+vb+v^{2}c\in R.$

The ring $R[x,\theta ]=\{a_{0}+a_{1}x+...+a_{n-1}x^{n-1}$ : $a_{i}\in R,$ $%
n\in N\}$ is called a skew polynomial ring. This ring is a non-commutative
ring. The addition in the ring $R[x,\theta ]$ is the usual polynomial
addition and multiplication is defined using the rule, $(ax^{i})(bx^{j})=a%
\theta ^{i}(b)x^{i+j}$. Note that $\theta ^{2}(a)=a$ for all $a\in R$. This
implies that $\theta $ is a ring automorphism of order $2$.

\begin{definition}
A subset $C$ of $R^{n}$ is callled a skew cyclic code of length $n$ if $C$
satisfies the following conditions,
\end{definition}

$i)$ $C$ is a submodule of $R^{n}$,

$ii)$ If $c=\left( c_{0},c_{1},...,c_{n-1}\right) \in C$, then $\sigma
_{\theta }\left( c\right) =\left( \theta (c_{n-1)},\theta (c_{0}),...,\theta
(c_{n-2})\right) \in C.$

Let $(f(x)+(x^{n}-1))$ be an element in the set $R_{n}=R\left[ x,\theta %
\right] /(x^{n}-1)$ and let $r(x)\in R\left[ x,\theta \right] $. Define
multiplication from left as follows,

\begin{equation*}
r(x)(f(x)+(x^{n}-1))=r(x)f(x)+(x^{n}-1)
\end{equation*}%
for any $r(x)\in R\left[ x,\theta \right] $.

\begin{theorem}
$R_{n}$ is a left $R\left[ x,\theta \right] $-module where multiplication
defined as in above.
\end{theorem}

\begin{theorem}
A code $C$ in $R_{n}$ is a skew cyclic code if and only if $C$ is a left $R%
\left[ x,\theta \right] $-submodule of the left $R\left[ x,\theta \right] $%
-module $R_{n}$.
\end{theorem}

\begin{theorem}
Let $C$ be a skew cyclic code in $R_{n}$ and let $f(x)$ be a polynomial in $%
C $ of minimal degree. If $f(x)$ is monic polynomial, then $C=\left(
f(x)\right) $ where $f(x)$ is a right divisor of $(x^{n}-1)$.
\end{theorem}

\begin{theorem}
A module skew cyclic code of length $n$ over $R$ is free iff it is generated
by a monic right divisor $f(x)$ of $x^{n}-1.$ Moreover, the set $%
\{f(x),xf(x),x^{2}f(x)$\newline
$,...,x^{n-\deg (f(x))-1}f(x)\}$ forms a basis of $C$ and the rank of $C$ is 
$n-\deg (f(x)).$
\end{theorem}

\begin{theorem}
Let $n$ be odd and $C$ be a skew cyclic code of length $n$. Then $C$ is
equivalent to cyclic code of length $n$ over $R$.
\end{theorem}

\begin{proof}
Since $n$ is odd, $gcd(2,n)=1.$ Hence there exist integers $b,c$ such that $%
2b+nc=1.$ So $2b=1-nc=1+zn$ where $z>0.$ Let $%
a(x)=a_{0}+a_{1}x+...+a_{n-1}x^{n-1}$ be a codeword in $C.$ Note that $%
x^{2b}a(x)=\theta ^{2b}(a_{0})x^{1+zn}+\theta
^{2b}(a_{1})x^{2+zn}+...+\theta
^{2b}(a_{n-1})x^{n+zn}=a_{n-1}+a_{0}x+...+a_{n-2}x^{n-2}\in C.$ \ Thus $C$
is a cyclic code of length\ $n$.
\end{proof}

\begin{corollary}
Let $n$ be odd. Then the number of distinct \ skew cyclic codes of length $n$
over $R$ is equal to the number of ideals in $R\left[ x\right] /(x^{n}-1)$
because of Theorem 44. If $x^{n}-1=\tprod\limits_{i=0}^{r}p_{i}^{s_{i}}(x)$
where $p_{i}(x)$ are irreducible polynomials over $Z_{3}$. Then the number
of distinct skew cyclic codes of length $n$ over $R$ is $\tprod%
\limits_{i=0}^{r}(s_{i}+1)^{3}.$
\end{corollary}

\begin{example}
Let $n=12$ and $f(x)=x^{3}+x^{2}+x+1.$ Then $f(x)$ generates a skew cyclic
codes of length $12$. This code is equivalent to a cyclic code of length $12$%
. Since $x^{12}-1=\left( x-1\right) ^{3}\left( x^{3}+x^{2}+x+1\right) ^{3}$,
it follows that there are $4^{6}$ skew cyclic code of length $12$.
\end{example}

\begin{definition}
A subset $C$ of $R^{n}$ is called a skew quasi-cyclic code of length $n$ if $%
C$ satisfies the following conditions,
\end{definition}

\ $i)$ $C$ is a submodule of $R^{n}$,

$ii)$ If $e=\left(
e_{0,0},...,e_{0,l-1},e_{1,0},...,e_{1,l-1},...,e_{s-1,0},..,e_{s-1,l-1}%
\right) \in C$, then \newline
$\tau _{\theta ,s,l}\left( e\right) =\left( \theta (e_{s-1,0}),...,\theta
(e_{s-1,l-1}),\theta (e_{0,0}),...,\theta (e_{0,l-1}),...,\theta
(e_{s-2,0}),...,\theta (e_{s-2,l-1})\right) $\newline
$\in C$.

We note that $x^{s}-1$ is a two sided ideal in $R\left[ x,\theta \right] $
if $m|s$ where $m$ is the order of $\theta $ and equal to two. So $R\left[
x,\theta \right] /(x^{s}-1)$ is well defined.

The ring $R_{s}^{l}=(R\left[ x,\theta \right] /(x^{s}-1))^{l}$ is a left $%
R_{s}=R\left[ x,\theta \right] /(x^{s}-1)$ module by the following
multiplication on the left \newline
$f(x)(g_{1}(x),...,g_{l}(x))=(f(x)g_{1}(x),...f(x)g_{l}(x))$. If the map $%
\gamma $ is defined by

\begin{equation*}
\gamma :R^{n}\longrightarrow R_{s}^{l}
\end{equation*}

$\left(
e_{0,0},...,e_{0,l-1},e_{1,0},...,e_{1,l-1},...,e_{s-1,0},...,e_{s-1,l-1}%
\right) \mapsto $ $(c_{0}(x),...,c_{l-1}(x))$ such that $e_{j}(x)=\tsum%
\limits_{i=0}^{s-1}e_{i,j}x^{i}\in R_{s}^{l}$ where $j=0,1,...,l-1$ then the
map $\gamma $ gives a one to one correspondence $R^{n}$ and the ring $%
R_{s}^{l}$.

\begin{theorem}
A subset $C$ of $R^{n}$ is a skew quasi-cyclic code of length $n=sl$ and
index $l$ if and only if $\gamma (C)$ is a left $R_{s}$-submodule of $%
R_{s}^{l}$.
\end{theorem}

A code $C$ is said to be skew constacyclic if $C$ is closed the under the
skew constacyclic shift $\sigma _{\theta ,\lambda }$ from $R^{n}$ to $R^{n}$
defined by $\sigma _{\theta ,\lambda }\left( \left(
c_{0},c_{1},...,c_{n-1}\right) \right) =\left( \theta \left( \lambda
c_{n-1}\right) ,\theta \left( c_{0}\right) ,...,\theta \left( c_{n-2}\right)
\right) .$

Privately, such codes are called skew cyclic and skew negacyclic codes when $%
\lambda $ is $1$ and $-1$, respectively.

\begin{theorem}
A code $C$ of length $n$ over $R$ is skew constacyclic iff the skew
polynomial representation of $C$ is a left ideal in $R\left[ x,\theta \right]
/(x^{n}-\lambda ).$
\end{theorem}

\section{The Gray Images of Skew Codes Over $R$}

\begin{proposition}
Let $\sigma _{\theta }$ be the skew cyclic shift on $R^{n}$, let $\phi $ be
the Gray map from $R^{n}$ to $Z_{3}^{3n}$ and let $\varphi $ be as in the
preliminaries. Then $\phi \sigma _{\theta }=\rho \varphi \phi $ where $\rho
(x,y,z)=(x,z,y)$ for every $x,y,z\in Z_{3}^{n}$.
\end{proposition}

\begin{proof}
Let $r_{i}=a_{i}+vb_{i}+v^{2}c_{i}$ be the elements of $R$, for $%
i=0,1,....,n-1.$ We have $\sigma _{\theta }\left(
r_{0},r_{1},...,r_{n-1}\right) =\left( \theta (r_{n-1}),\theta
(r_{0}),...,\theta (r_{n-2})\right) .$ If \ we apply $\phi $, we have 
\begin{eqnarray*}
\phi \left( \sigma _{\theta }\left( r_{0},...,r_{n-1}\right) \right) &=&\phi
(\theta (r_{n-1}),\theta (r_{0}),...,\theta (r_{n-2})) \\
&=&(a_{n-1},...,a_{n-2},a_{n-1}+2b_{n-1}+c_{n-1},...,a_{n-2}+2b_{n-2} \\
&&+c_{n-2},a_{n-1}+b_{n-1}+c_{n-1},...,a_{n-2}+b_{n-2}+c_{n-2})
\end{eqnarray*}%
On the other hand, $\phi
(r_{0},...,r_{n-1})=(a_{0},...,a_{n-1},a_{0}+b_{0}+c_{0},...,a_{n-1}+b_{n-1}+c_{n-1},a_{0}+2b_{0}+c_{0},...,a_{n-1}+2b_{n-1}+c_{n-1}) 
$. If we apply $\varphi ,$ we have

$\varphi \left( \phi \left( r_{0},r_{1},...,r_{n-1}\right) \right)
=(a_{n-1},...,a_{n-2},a_{n-1}+b_{n-1}+c_{n-1},...,a_{n-2}+b_{n-2}+c_{n-2},a_{n-1}+2b_{n-1}+c_{n-1},...,a_{n-2}+2b_{n-2}+c_{n-2}) 
$. If we apply $\rho $, we have $\rho (\varphi \left( \phi \left(
r_{0},r_{1},...,r_{n-1}\right) \right)
)=(a_{n-1},...,a_{n-2},a_{n-1}+2b_{n-1}+c_{n-1},...,a_{n-2}+2b_{n-2}+c_{n-2},a_{n-1}+b_{n-1}+c_{n-1},...,a_{n-2}+b_{n-2}+c_{n-2}) 
$. So, we have $\phi \sigma _{\theta }=\rho \varphi \phi $.
\end{proof}

\begin{theorem}
The Gray image a skew cyclic code over $R$ of length $n$ is permutation
equivalent to quasi-cyclic code of index $3$ over $Z_{3}$ with length $3n$.
\end{theorem}

\begin{proof}
Let $C$ be a skew cyclic codes over $S$ of length $n$. That is $\sigma
_{\theta }(C)=C$. If we apply $\phi $, we have $\phi (\sigma _{\theta
}(C))=\phi (C)$. From the Proposition $48$, $\phi (\sigma _{\theta
}(C))=\phi (C)=$ $\rho (\varphi (\phi (C)))$. So, $\phi (C)$ is permutation
equivalent to quasi-cyclic code of index $3$ over $Z_{3}$ with length $3n$.
\end{proof}

\begin{proposition}
Let $\tau _{\theta ,s,l}$ be skew quasi-cyclic shift on $R^{n}$, let $\phi $
be the Gray map from $R^{n}$ to $Z_{3}^{3n}$, let $\Gamma $ be as in the
preliminaries, let $\rho $ be as above. Then $\phi \tau _{\theta ,s,l}=\rho
\Gamma \phi $.
\end{proposition}

\begin{theorem}
The Gray image a skew quasi-cyclic code over $R$ of length $n$ with index $l$
is permutation equivalent to $l$ quasi-cyclic code of index $3$ over $Z_{3}$
with length $3n$.
\end{theorem}

\begin{proposition}
Let $\sigma _{\theta ,\lambda }$ be skew constacyclic shift on $R^{n}$, let $%
\phi $ be the Gray map from $R^{n}$ to $Z_{3}^{3n}$, let $\rho $ be as
above. Then $\phi \nu =\rho \phi \sigma _{\theta ,\lambda }$.
\end{proposition}

\begin{theorem}
The Gray image a skew constacyclic code over $R$ of length $n$ is
permutation equivalent to a constacyclic code over $Z_{3}$ with length $3n$.
\end{theorem}

The proof of Proposition $52$, $54$ and Theorem $53$, $55$ are similiar to
the proof Proposition $50$ and Theorem $51$.

\section{Quasi-constacyclic and Skew Quasi-constacyclic Codes over $R$}

Let $M_{s}=R\left[ x\right] /\left\langle x^{s}-\lambda \right\rangle $
where $\lambda $ is a unit element of $R.$

\begin{definition}
A subset $C$ of $R^{n}$ is a called a quasi-constacyclic code of length $%
n=ls $ with index $l$ if
\end{definition}

$i)$ $C$ is a submodule of $R^{n},$

$ii)$ if $e=\left(
e_{0,0},...,e_{0,l-1},e_{1,0},...,e_{1,l-1},...,e_{s-1,0},...,e_{s-1,l-1}%
\right) \in C$ then\newline
$\nabla _{\lambda ,l}\left( e\right) =\left( \lambda e_{s-1,0},...,\lambda
e_{s-1,l-1},e_{0,0},...,e_{0,l-1},e_{1,0},...,e_{1,l-1},...,e_{s-2,0},...,e_{s-2,l-1}\right) \in C. 
$

When $\lambda =1$ the quasi-constacyclic codes are just quasi-cyclic codes.

Since $x^{s}-\lambda =f_{1}(x)f_{2}(x)...f_{r}(x),$ it follows that

\begin{equation*}
(R\left[ x\right] /(x^{s}-\lambda ))^{l}\cong (R\left[ x\right]
/(f_{1}(x)))^{l}\times (R\left[ x\right] /(f_{2}(x)))^{l}\times ...\times (R%
\left[ x\right] /(f_{r}(x)))^{l}.
\end{equation*}

Every submodule of $(R\left[ x\right] /(x^{s}-\lambda ))^{l}$ is a direct
product of submodules of $(R\left[ x\right] /(f_{t}(x)))^{l}$ for $1\leq
t\leq r.$

\begin{theorem}
If $(s,3)=1$ then a quasi-constacyclic code of length $n=sl$ with index $l$
over $R$ is a direct product of linear codes over $R\left[ x\right]
/(f_{t}(x))$ for $1\leq t\leq r.$
\end{theorem}

Let $x^{s}-\lambda =f_{1}(x)f_{2}(x)...f_{r}(x)$ be the factorization of $%
x^{s}-\lambda $ into irreducible polynomials.Thus, if $(s,3)=1$ and $C_{i}$
is a linear code of length $l$ over $R\left[ x\right] /(f_{t}(x))$ for $%
1\leq t\leq r,$ then $\tprod\limits_{t=1}^{r}C_{t}$ is a quasi-constacyclic
code of length $n=sl$ over $R$ with $\tprod\limits_{t=1}^{r}\left\vert
C_{t}\right\vert $ codewords.

Define a map $\chi :R^{n}\rightarrow M_{s}^{l}$ by $\chi \left( e\right)
=\left( e_{0}\left( x\right) ,e_{1}\left( x\right) ,...,e_{l-1}\left(
x\right) \right) $ where $e_{j}\left( x\right)
=\sum_{i=o}^{s-1}e_{ij}x^{i}\in M_{s}$, $j=0,1,...,l-1.$

\begin{lemma}
Let $\chi \left( C\right) $ denote the image of $C$ under $\chi .$ The map $%
\chi $ induces a one to one correspondence between quasi-constacyclic codes
over $R$ of length $n$ with index $l$ and linear codes over $M_{s}$ of
length $l.$
\end{lemma}

We define a conjugation map on $M_{s}$ as one that acts as the identity on
the elements of $R$ and that sends $x$ to $x^{-1}=x^{s-1}$, and extended
linearly.

We define on $R^{n=sl}$ the usual Euclidean inner product for

\begin{equation*}
e=\left(
e_{0,0},...,e_{0,l-1},e_{1,0},...,e_{1,l-1},...,e_{s-1,0},...,e_{s-1,l-1}%
\right)
\end{equation*}

and

\begin{equation*}
c=\left(
c_{0,0},...,c_{0,l-1},c_{1,0},...,c_{1,l-1},...,c_{s-1,0},...,c_{s-1,l-1}%
\right)
\end{equation*}%
we define $e.c=\sum_{i=0}^{s-1}\sum_{j=0}^{l-1}e_{ij}c_{ij}.$

On $M_{s}^{l}$, we define the Hermitian inner product \ for $a(x)=\left(
a_{0}(x),a_{1}(x),...,a_{l-1}(x)\right) $ and $b(x)=\left(
b_{0}(x),b_{1}(x),...,b_{l-1}(x)\right) ,$

\begin{equation*}
\left\langle a,b\right\rangle =\sum_{j=0}^{l-1}a_{j}(x)\overline{b_{j}(x)}.
\end{equation*}

\begin{theorem}
Let $e,c\in R^{n}.$ Then $\left( \nabla _{\lambda ,l}^{k}\left( e\right)
\right) .c=0$ for all $k=0,...,s-1$ iff $\left\langle \chi \left( e\right)
,\chi \left( c\right) \right\rangle =0.$
\end{theorem}

\begin{corollary}
Let $C$ be a quasi-constacyclic code of length $sl$ with index $l$ over $R$
and $\chi \left( C\right) $ be its image in $M_{s}^{l}$ under $\chi .$ Then $%
\chi \left( C\right) ^{\bot }=\chi \left( C^{\bot }\right) ,$ where the dual
in $R^{sl}$ is taken with respect to the Euclidean inner product, while the
dual in $M_{s}^{l}$ is taken with respect to the Hermitian inner product.
The dual of a quasi-constacyclic code of length $sl$ with index $l$ over $R$
is a quasi-constacyclic code of length $sl$ with index $l$ over $.$
\end{corollary}

From [27] we get the following results.

\begin{theorem}
Let C be a quasi-constacyclic code of length n=sl with index l over R. Let $%
C^{\bot }$ is the dual of C. If $C=C_{1}\oplus C_{2}\oplus ...\oplus C_{r}$
then $C^{\bot }=C_{1}^{\bot }\oplus C_{2}^{\bot }\oplus ...\oplus
C_{r}^{\bot }.$
\end{theorem}

\begin{theorem}
Let $C=C_{1}\oplus C_{2}\oplus ...\oplus C_{r}$ be a quasi-constacyclic code
of length n=sl with index l over R where $C_{t}$ is a free linear code of
length l with rank $k_{t}$ over $R\left[ x\right] /(f_{t}(x))$for $1\leq
t\leq r.$ Then C is a $\kappa $-generator quasi-constacyclic code and $%
C^{\bot }$ is an $(l-\kappa ^{\prime })$-generator quasi-constacyclic code
where $\kappa =\max_{t}(k_{t})$ and $\kappa ^{\prime }=\min_{t}(k_{t}).$
\end{theorem}

Let $M_{\theta ,s}=R\left[ x,\theta \right] /\left\langle x^{s}-\lambda
\right\rangle $ where $\lambda $ is a unit element of $R.$ Let $\theta $ be
an automorphism of $R$ with $\left\vert \left\langle \theta \right\rangle
\right\vert =m=2.$

\begin{definition}
A subset $C$ of $R^{n}$ is a called a skew quasi-constacyclic code of length 
$n=ls,m|s,$ with index $l$ if
\end{definition}

$i)$ $C$ is a submodule of $R^{n},$

$ii)$ if $e=\left(
e_{0,0},...,e_{0,l-1},e_{1,0},...,e_{1,l-1},...,e_{s-1,0},...,e_{s-1,l-1}%
\right) \in C$ then\newline
$\nabla _{\theta ,\lambda ,l}\left( e\right) =(\theta \left( \lambda
e_{s-1,0}\right) ,...,\theta \left( \lambda e_{s-1,l-1}\right) ,\theta
\left( e_{0,0}\right) ,...,\theta \left( e_{0,l-1}\right) ,\theta \left(
e_{1,0}\right) ,...,$\newline
$\theta \left( e_{1,l-1}\right) ,...,\theta \left( e_{s-2,0}\right)
,...,\theta \left( e_{s-2,l-1}\right) )\in C.$

When $\lambda =1$ the skew quasi-constacyclic codes are just skew
quasi-cyclic codes.

The ring $M_{\theta ,s}^{l}$ is a left $M_{\theta ,s}$ module where we
define multiplication from left by $%
f(x)(g_{1}(x),...,g_{l}(x))=(f(x)g_{1}(x),...f(x)g_{l}(x)).$

Define a map $\Lambda :R^{n}\rightarrow M_{\theta ,s}^{l}$ by $\Lambda
\left( e\right) =\left( e_{0}\left( x\right) ,e_{1}\left( x\right)
,...,e_{l-1}\left( x\right) \right) $ where $e_{j}\left( x\right)
=\sum_{i=o}^{s-1}e_{ij}x^{i}\in M_{\theta ,s}$, $j=0,1,...,l-1.$

\begin{lemma}
Let $\Lambda \left( C\right) $ denote the image of $C$ under $\Lambda .$ The
map $\Lambda $ induces a one to one correspondence between
quasi-constacyclic codes over $R$ of length $n$ with index $l$ and linear
codes over $M_{\theta ,s}$ of length $l.$

\begin{theorem}
A subset $C$ of $R^{n}$ is a skew quasi-constacyclic code of length $n=ls$
with index $l$ iff is a left submodule of the ring $M_{\theta ,s}^{l}.$

\begin{proof}
Let $C$ be a skew quasi-constacyclic code of index $l$ over $R.$Suppose that 
$\Lambda \left( C\right) $ forms a submodule of $M_{\theta ,s}^{l}$. $%
\Lambda \left( C\right) $ is closed under addition and scalar
multiplication. Let $\Lambda \left( e\right) =\left( e_{0}\left( x\right)
,e_{1}\left( x\right) ,...,e_{l-1}\left( x\right) \right) \in \Lambda \left(
C\right) $ for $e=\left(
e_{0,0},...,e_{0,l-1},e_{1,0},...,e_{1,l-1},...,e_{s-1,0},...,e_{s-1,l-1}%
\right) \in C.$ Then $x\Lambda \left( e\right) \in \Lambda \left( C\right) .$
By linearity it follows that $r\left( x\right) \Lambda \left( e\right) \in
\Lambda \left( C\right) $ for any $r(x)\in M_{\theta ,s}.$ Therefore, $%
\Lambda \left( C\right) $ is a left module of $M_{\theta ,s}^{l}.$

Conversely, suppose $E$ is an $M_{\theta ,s}$ left submodule of $M_{\theta
,s}^{l}.$ Let $C=\Lambda ^{-1}\left( E\right) =\left\{ e\in R^{n}:\Lambda
\left( e\right) \in E\right\} .$ We claim that $C$ is a skew
quasi-constacyclic code of $R.$ Since $\Lambda $ is a isomorphism, $C$ is a
linear code of length $n$ over $R.$ Let $e=\left(
e_{0,0},...,e_{0,l-1},e_{1,0},...,e_{1,l-1},...,e_{s-1,0},...,e_{s-1,l-1}%
\right) \in C.$ Then $\Lambda \left( e\right) =(e_{0}\left( x\right) ,$%
\newline
$e_{1}\left( x\right) ,...,e_{l-1}\left( x\right) )\in \Lambda \left(
C\right) ,$ where where $e_{j}\left( x\right)
=\sum_{i=o}^{s-1}e_{ij}x^{i}\in M_{\theta ,s}$ for $j=0,1,...,l-1.$ It is
easy to see that $\Lambda \left( \nabla _{\theta ,\lambda ,l}\left( e\right)
\right) =x\left( e_{0}\left( x\right) ,e_{1}\left( x\right)
,...,e_{l-1}\left( x\right) \right) $\newline
$=\left( xe_{0}\left( x\right) ,xe_{1}\left( x\right) ,...,xe_{l-1}\left(
x\right) \right) \in E.$ Hence $\nabla _{\theta ,\lambda ,l}\left( e\right)
\in C.$ So, $C$ is a skew quasi-constacyclic code $C$.
\end{proof}
\end{theorem}
\end{lemma}

On $R^{n=sl}$ the usual Euclidean inner product for

\begin{equation*}
e=\left(
e_{0,0},...,e_{0,l-1},e_{1,0},...,e_{1,l-1},...,e_{s-1,0},...,e_{s-1,l-1}%
\right)
\end{equation*}

and

\begin{equation*}
c=\left(
c_{0,0},...,c_{0,l-1},c_{1,0},...,c_{1,l-1},...,c_{s-1,0},...,c_{s-1,l-1}%
\right)
\end{equation*}%
we define $e.c=\sum_{i=0}^{s-1}\sum_{j=0}^{l-1}e_{ij}c_{ij}.$

We define a conjugation map $\Omega $ on $M_{\theta ,s}^{l}$ such that $%
\Omega (cx^{i})=\theta ^{-1}(c)x^{s-1},0\leq i\leq s-1,$ and extended
linearly. We define the Hermitian inner product \ for $a=\left(
a_{0}(x),a_{1}(x),...,a_{l-1}(x)\right) $ and $b=\left(
b_{0}(x),b_{1}(x),...,b_{l-1}(x)\right) ,$

\begin{equation*}
\left\langle a,b\right\rangle =\sum_{j=0}^{l-1}a_{j}(x)\Omega \left(
b_{j}(x)\right) .
\end{equation*}

\begin{theorem}
Let $e,c\in R^{n}.$ Then $\left( \nabla _{\theta ,\lambda ,l}^{k}\left(
e\right) \right) .c=0$ for all $k=0,...,s-1$ iff $\left\langle \Lambda
\left( e\right) ,\Lambda \left( c\right) \right\rangle =0.$
\end{theorem}

\begin{proof}
Since $\theta ^{s}=1,\left\langle e,c\right\rangle =0$ is equivalent to

\begin{eqnarray*}
0 &=&\sum_{j=0}^{l-1}e_{j}(x)\Omega \left( c_{j}(x)\right)
=\sum_{j=0}^{l-1}\left( \tsum\limits_{i=0}^{s-1}e_{ij}x^{i}\right) \Omega
\left( \tsum\limits_{k=0}^{s-1}c_{kj}x^{k}\right) \\
&=&\sum_{j=0}^{l-1}\left( \tsum\limits_{i=0}^{s-1}e_{ij}x^{i}\right) \left(
\tsum\limits_{k=0}^{s-1}\theta ^{-1}(c_{kj})x^{s-k}\right) \\
&=&\sum_{j=0}^{l-1}\left(
\sum_{j=0}^{l-1}\tsum\limits_{i=0}^{s-1}e_{i+h,j}\theta ^{h}(c_{ij})\right)
x^{h}
\end{eqnarray*}%
where the subscript $i+h$ is taken modulo $s$. Equating the coefficients of $%
x^{h}$ on both sides, we have $\sum_{j=0}^{l-1}\tsum%
\limits_{i=0}^{s-1}w_{i+h,j}\theta ^{h}(c_{ij})=0,$ for all $0\leq h\leq
s-1. $ $\sum_{j=0}^{l-1}\tsum\limits_{i=0}^{s-1}e_{i+h,j}\theta
^{h}(c_{ij})=0$ is equivalent to $\theta ^{h}(\nabla _{\theta ,\lambda
,l}^{s-h}\left( e\right) .c)=0$ which is further equivalent to $\nabla
_{\theta ,\lambda ,l}^{s-h}(e,c)=0,$ for all $0\leq h\leq s-1.$ Since $0\leq
h\leq s-1,$ \ condition is equivalent to $\left( \nabla _{\theta ,\lambda
,l}^{k}\left( e\right) \right) .c=0$ for all $k=0,...,s-1.$
\end{proof}

\begin{corollary}
Let $C$ be a skew quasi-constacyclic code of length $n=sl$ with index $l$
over $R$. Then $C^{\bot }=\left\{ a(x)\in M_{\theta ,s}^{l}\text{ }:\text{ }%
\left\langle a(x),b(x)\right\rangle =0,\text{ }\forall \text{ }b(x)\in
C\right\} .$
\end{corollary}

\begin{corollary}
Let $C$ be a skew quasi-constacyclic code of length $sl$ with index $l$ over 
$R$ and $\Lambda \left( C\right) $ be its image in $M_{\theta ,s}^{l}$ under 
$\Lambda .$ Then $\Lambda \left( C\right) ^{\bot }=\Lambda \left( C^{\bot
}\right) ,$ where the dual in $R^{sl}$ is taken with respect to the
Euclidean inner product, while the dual in $M_{\theta ,s}^{l}$ is taken with
respect to the Hermitian inner product. The dual of a skew
quasi-constacyclic code of length $sl$ with index $l$ over $R$ is a skew
quasi-constacyclic code of length $sl$ with index $l$ over $\ R.$
\end{corollary}

\begin{proposition}
Let $\nabla _{\theta ,\lambda ,l}$ be skew quasi-constacyclic shift on $%
R^{n} $, let $\phi $ be the Gray map from $R^{n}$ to $Z_{3}^{3n}$. Then $%
\phi \nabla _{\lambda ,l}=\rho \phi \nabla _{\theta ,\lambda ,l},$ where $%
\rho (x,y,z)=(x,z,y)$ for every $x,y,z\in Z_{3}^{n}.$
\end{proposition}

\begin{proof}
The proof is similar to the proof of Proposition $50$.
\end{proof}

\begin{theorem}
The Gray image a skew quasi-constacyclic code over $R$ of length $n$ is
permutation equivalent to a quasi-constacyclic code over $Z_{3}$ with length 
$3n$.
\end{theorem}

\begin{proof}
The proof is similar to the proof of Theorem $51$.
\end{proof}

\section{1-generator Skew Quasi-constacyclic Codes over $R$}

A $1$-generator skew quasi-constacyclic code over $R$ is a left $M_{\theta
,s}$-submodule of $M_{\theta ,s}^{l}$ generated by $\mathbf{f(x)}=\left(
f_{1}(x),f_{2}(x),...,f_{l}(x)\right) \in M_{\theta ,s}^{l}$ has the form $%
C=\left\{ g(x)\left( f_{1}(x),f_{2}(x),...,f_{l}(x)\right) :\text{ }g(x)\in
M_{\theta ,s}\right\} .$

Define the following map

\begin{equation*}
\Pi _{i}:M_{\theta ,s}^{l}\longrightarrow M_{\theta ,s}
\end{equation*}%
defined by $\left( e_{1}\left( x\right) ,e_{2}\left( x\right)
,...,e_{l}\left( x\right) \right) \longmapsto e_{i}(x),$ $1\leq i\leq l.$
Let $\Pi _{i}(C)=C_{i}.$ Since $C$ is a left $M_{\theta ,s}$-submodule of $%
M_{\theta ,s}^{l},$ $C_{i}$ is a left $M_{\theta ,s}$-submodule of $%
M_{\theta ,s},$that is a left ideal of $M_{\theta ,s}.$ $C_{i}$ is generated
by $f_{i}(x).$ Hence $C_{i}$ is a principal skew constacyclic code of length 
$n$ over $R$. $f_{i}(x)$ is a monic right divisor of $x^{s}-\lambda $ that
is $x^{s}-\lambda =h_{i}(x)f_{i}(x),$ $1\leq i\leq l.$

A generator of $C$ has the form $\mathbf{f(x)}=\left(
g_{1}(x)f_{1}(x),g_{2}(x)f_{2}(x),...,g_{l}(x)f_{l}(x)\right) $ where $%
g_{i}(x)\in R[x,\theta ]$ such that $g_{i}(x)$ and $h_{i}(x)$ are right
coprime for all $1\leq i\leq l.$

\begin{definition}
Let $C=\left( g_{1}(x)f_{1}(x),g_{2}(x)f_{2}(x),...,g_{l}(x)f_{l}(x)\right) $
be a skew quasi-constacyclic code of length $n=sl$ with index $l$. Then
unique monic polynomial
\end{definition}

\begin{equation*}
f(x)=gcld(\mathbf{f(x),}x^{s}-\lambda
)=gcld(f_{1}(x),f_{2}(x),...,f_{l}(x),x^{s}-\lambda )
\end{equation*}%
is called the generator polynomial of $C$.

\begin{theorem}
Let $C$ be a $1$-generator skew quasi-constacyclic code of length $n=s$l
with index $l$ over $R$ generated by $\mathbf{f(x)}=\left(
f_{1}(x),f_{2}(x),...,f_{l}(x)\right) $ where $f_{i}(x)$ is a monic divisor
of $x^{s}-\lambda .$ Then $C$ is a $R$-free code with rank $s-deg(f(x))$
where $f(x)=gcld(\mathbf{f(x),}x^{s}-\lambda ).$ Moreover, the set $\left\{ 
\mathbf{f(x)},x\mathbf{f(x)},...,x^{n-\deg (f(x))-1}\mathbf{f(x)}\right\} $
forms an $R$-basis of $C.$
\end{theorem}

\begin{proof}
Since $gcld(f_{i}(x),x^{s}-\lambda )=m_{i}(x),$ it follows that $%
f(x)=gcld(m_{1}(x),m_{2}(x)$\newline
$,...,m_{l}(x))$ where $\Pi _{i}(C)=(f_{i}(x))=(m_{i}(x))$ with $%
m_{i}(x)|(x^{s}-\lambda )$ for all $1\leq i\leq l.$ Let $c(x)=\tsum%
\limits_{i=0}^{n-k-1}c_{i}x^{i}$ and $c(x)\mathbf{f(x)}=0.$ Then $%
(x^{s}-\lambda )|c(x)f_{i}(x)$ for all $1\leq i\leq l.$ Hence $%
(x^{s}-\lambda )|c(x)f_{i}(x)c_{i}(x)$ with $gcld(c_{i}(x),\frac{%
x^{s}-\lambda }{f_{i}(x)})=1.$ That is $\frac{x^{s}-\lambda }{f_{i}(x)}|c(x)$
\ which implies that $\frac{x^{s}-\lambda }{f(x)}|c(x).$ Since $\deg (\frac{%
x^{s}-\lambda }{f(x)})=s-k>\deg (c(x))=n-k-1,$ it is follows that $c(x)=0.$
Thus, $\mathbf{f(x)},x\mathbf{f(x)},...,x^{n-\deg (f(x))-1}\mathbf{f(x)}$
are $R$-linear independent. Further, $\mathbf{f(x)},x\mathbf{f(x)}%
,...,x^{n-\deg (f(x))-1}\mathbf{f(x)}$ generate $C$. So, $\left\{ \mathbf{%
f(x)},x\mathbf{f(x)},...,x^{n-\deg (f(x))-1}\mathbf{f(x)}\right\} $ forms an 
$R$-basis of $C$.
\end{proof}

\section{References}

$\ \ \ \left[ 1\right] $\ A. Bayram, I. \c{S}iap, Structure of codes over
the ring $Z_{3}[v]/\left\langle v^{3}-v\right\rangle ,$ AAECC,DOI $%
10.1007/s00200-013-0208-x$, 2013.

$\left[ 2\right] $ A. Dertli, Y. Cengellenmis, S. Eren, On quantum codes
obtained from cyclic codes over $A_{2},$ Int. J. Quantum Inform., vol. $13$, 
$2(2015)$ $1550031.$

$[3]$ A. Dertli, Y. Cengellenmis, S. Eren, Quantum codes over the ring $%
F_{2}+uF_{2}+u^{2}F_{2}+...+u^{m}F_{2},$ Int. Journal of Alg., vol. 9,
3(2015), 115 - 121.

$[4]$ A. M. Steane,Simple quantum error correcting codes, Phys. Rev. A, $%
54\left( 1996\right) ,$ $4741-4751.$

$\left[ 5\right] $ A. R. Calderbank, E.M.Rains, P.M.Shor, N.J.A.Sloane,
Quantum error correction via codes over $GF\left( 4\right) ,$ IEEE Trans.
Inf. Theory, $44\left( 1998\right) ,$ $1369-1387.$

$[6]$ A. R. \ Hammons, V. Kumar, A. R. Calderbank, N. J. A. Sloane, P. Sole,
The $Z_{4}$-linearity of Kerdock, Preparata, Goethals and related codes,
IEEE Trans. Inf. Theory 40(1994) 301-319.

$[7]$ D. Boucher, W. Geiselmann, F. Ulmer, Skew cyclic codes, Appl. Algebra.
Eng.Commun Comput., Vol. $18$, No. $4$,$2007$, $379-389$.

$\left[ 8\right] $ D. Boucher, P. Sole, F. Ulmer, Skew constacyclic codes
over Galois rings, Advance of Mathematics of Communications, Vol. $2$,
Number $3$, $2008$, $273-292$.

$\left[ 9\right] $ D. Boucher, F. Ulmer, Coding with skew polynomial rings,
Journal of Symbolic Computation, $44$, $2009$, $1644-1656$.

$\left[ 10\right] $ I. Siap, T. Abualrub, N. Ayd\i n, P. Seneviratne, Skew
cyclic codes of arbitrary length, Int. Journal of Information and Coding
Theory, $2010$.

$\left[ 11\right] $ J. F. Qian, L. N. Zhang, S. X. Zhu, $\left( 1+u\right) $%
-constacyclic and cyclic codes over $F_{2}+uF_{2}$, Applied Mathematics
Letters, $19(2006)820-823$.

$\left[ 12\right] $ J. Gao, L. Shen, F. W. Fu, Skew generalized quasi-cyclic
codes over finite fields, arXiv: $1309$,$1621v1$.

$\left[ 13\right] $ J.Qian, Quantum codes from cyclic codes over $%
F_{2}+vF_{2},$ Journal of Inform.$\&$ computational Science $10:6\left(
2013\right) ,$ $1715-1722.$

$\left[ 14\right] $ J.Qian, W.Ma, W.Gou, Quantum codes from cyclic codes
over finite ring, Int. J. Quantum Inform., $7\left( 2009\right) ,$ $%
1277-1283.$

$[15]$ J. Gao, Skew cyclic codes over $F_{p}+vF_{p}$, J. Appl. Math. \&
Informatics, 31(2013), No.3-4, 337-342.

$[16]$ J. Gao, L. Shen, F. W. Fu, Skew Generalized Quasi-Cyclic Codes over
Finite Fields, arXiv:1309.1621v1.

$\left[ 17\right] $ M. Ashraf, G. Mohammad, Quantum codes from cyclic codes
over $F_{3}+vF_{3}$, International Journal of Quantum Information, vol. $12$%
, No. $6(2014)$ $1450042$.

$\left[ 18\right] $ M. Bhaintwal, Skew quasi-cyclic codes over Galois rings,
Des. Codes Cryptogr., DOI $10.1007/s10623-011-9494-0$.

$[19]$ M. Bhaintwal, S. K. Wasan, On quasi-cyclic codes over $Z_{q},$
AAECC,DOI $10.1007/s00200-009-0110-8$, (2009)20:459-480.

$[20]$ M. Grassl, T. Beth, On optimal quantum codes, International Journal
of Quantum Information, 2(2004) 55-64.

$\left[ 21\right] $ M. Wu, Skew cyclic and quasi-cyclic codes of arbitrary
length over Galois rings, International Journal of Algebra, vol $7$, $2013$,
no $17$, $803-807$.

$[22]$ Maheshanand, S. K. Wasan, On Quasi-cyclic Codes over Integer Residue
Rings, AAECC, Lecture Notes in Computer Science Volume 4851, 330-336, 2007.

$\left[ 23\right] $ P.W.Shor,Scheme for reducing decoherence in quantum
memory, Phys. Rev. A, $52\left( 1995\right) ,$ $2493-2496.$

$\left[ 24\right] $ S. Jitman, S. Ling, P. Udomkovanich, Skew constacyclic
codes over \ finite chain rings, AIMS Journal.

$[25]$ S. \ Ling, \ P. Sole, On the algebraic structures of quasi-cyclic
codes I: finite fields. IEEE Trans. Inf. Theory 47, 2751-2760 (2001).

$[26]$ S. \ Ling, \ P. Sole, On the algebraic structures of quasi-cyclic
codes II: chain rings. Des.Codes Cryptogr. 30, 113130 (2003).

$[27]$ S. \ Ling, \ P. Sole, On the algebraic structures of quasi-cyclic
codes III: generator theory. IEEE Trans. Inf. Theory 51, 2692-2000 (2005).

$[28]$ S. Zhu, L. Wang, A class of constacyclic codes over $F_{p}+vF_{p}$
and their Gray images, Discrete Math. 311, 2677-2682, 2011.

$\left[ 29\right] $ T. Abualrub, A. Ghrayeb, N. Ayd\i n, I. Siap, On the
construction of skew quasi-cyclic codes, IEEE Transsactions on Information
Theory, Vol $56$, No $5$, $2010$, $2081-2090$.

$\left[ 30\right] $ T. Abualrub, N. Ayd\i n, P. Seneviratne, On $\theta $%
-cyclic codes over $F_{2}+vF_{2}$, Australasian Journal of Combinatorics,
54(2012), 115-126.

$\left[ 31\right] $ X.Kai,S.Zhu, Quaternary construction bof quantum codes
from cyclic codes over $F_{4}+uF_{4},$ Int. J. Quantum Inform., $9\left(
2011\right) ,$ $689-700.$

$\left[ 32\right] $ X.Yin, W.Ma, Gray Map And Quantum Codes Over The Ring $%
F_{2}+uF_{2}+u^{2}F_{2},$ International Joint Conferences of IEEE
TrustCom-11, $2011.$

$\left[ 33\right] $ Y. Cengellenmis, A. Dertli, S.T. Dougherty, Codes over
an infinite family of rings with a Gray map, Designs, Codes and
Cryptography, $(2014)72$ :$559-580$

\end{document}